\documentclass[11pt,letterpaper]{article}
\usepackage[letterpaper,margin=1in]{geometry}
\usepackage{fontspec}
\usepackage{amsmath,amssymb,amsthm,mathtools,microtype}
\usepackage{booktabs,array,tabularx,enumitem,needspace}
\usepackage{placeins,float}
\usepackage{cite}
\usepackage{xurl,graphicx,tikz,xcolor}
\usetikzlibrary{arrows.meta,positioning,calc,decorations.pathreplacing}
\definecolor{inkblue}{HTML}{173E63}
\definecolor{lightblue}{HTML}{EDF4FA}
\definecolor{warmfill}{HTML}{FAE8D4}
\definecolor{greenfill}{HTML}{E2F0E7}
\definecolor{warmink}{HTML}{96541C}
\definecolor{greenink}{HTML}{2E644B}
\tikzset{
 figarrow/.style={-{Stealth[length=1.7mm]},draw=inkblue,semithick},
 datacell/.style={draw=inkblue!60,minimum width=.55cm,minimum height=.55cm,inner sep=1pt,font=\normalsize},
 smalltag/.style={font=\normalsize,text=inkblue}
}
\usepackage[colorlinks=true,linkcolor=inkblue,citecolor=inkblue,urlcolor=inkblue]{hyperref}
\usepackage{bookmark}
\hypersetup{pdftitle={Undirected Multiple-Unicast Network Coding and Faster External-Memory Integer Sorting},pdfauthor={Baochun Li, Zongpeng Li, Sirui Liu},pdfsubject={A finite coding separation and memory-dependent integer sorting bounds},pdflang={en-US}}
\setlist{itemsep=3pt,topsep=5pt}
\newtheoremstyle{paperplain}{7pt}{7pt}{\normalfont}{}{\bfseries}{.}{.5em}{}
\theoremstyle{paperplain}
\newtheorem{theorem}{Theorem}[section]
\newtheorem{lemma}[theorem]{Lemma}
\newtheorem{proposition}[theorem]{Proposition}
\newtheorem{corollary}[theorem]{Corollary}
\theoremstyle{definition}

\newcommand{\rank}{\operatorname{rank}}

\newcolumntype{Y}{>{\raggedright\arraybackslash}X}
\title{Undirected Multiple-Unicast Network Coding\\and Faster External-Memory Integer Sorting}

\author{
  \normalsize
  Baochun Li\textsuperscript{1}\qquad
  Zongpeng Li\textsuperscript{2,3}\qquad
  Sirui Liu\textsuperscript{3,2}\\[0.7em]
  {\small
  \begin{tabular}{@{}c@{}}
    \textsuperscript{1}\,University of Toronto, Toronto, Canada\\[0.15em]
    \textsuperscript{2}\,Quan Cheng Laboratory, Shandong, China\\[0.15em]
    \textsuperscript{3}\,Tsinghua University, Beijing, China
  \end{tabular}
  }
}

\date{}
\begin{document}
\pagenumbering{roman}
\maketitle
\thispagestyle{empty}
% The standard article abstract environment sets 10pt text.
% Use the 11pt body size required by the STOC review format.
\begin{center}\textbf{Abstract}\end{center}
\noindent In undirected multiple-unicast networks, where independent messages
are sent between specified source--destination pairs, network coding
allows intermediate nodes to combine messages, whereas fractional
routing splits and forwards individual messages along paths without
mixing them. Proposed in 2004 by Li and Li and independently by Harvey, Kleinberg, and Lehman, the multiple-unicast conjecture---long a
central open problem in network coding with deep connections to
computational complexity---asserts that network coding and fractional
routing have the same achievable rate region. We construct a
counterexample refuting this conjecture and use it to reduce block
transfers in external-memory data rearrangement, thereby accelerating
external-memory integer sorting and answering in the affirmative the
longstanding question of Aggarwal and Vitter (1988) on whether integer
operations can beat the comparison-sorting input/output (I/O) bound.

Specifically, we give a deterministic sorting algorithm: for $N$
integers of $\Theta(\log N)$ bits, a memory holding $M$ integers,
and block transfers of $B$ integers, where $2B\le M\le N$, the
number of block transfers is
\[
  O\!\left(\frac{N}{B}\left[1+
  \frac{\log N}{(\log(M/B))^{1-\varepsilon}(\log M)^\varepsilon}
  \right]\right),
\]
where $\varepsilon=1.5\times10^{-5}$ is an absolute constant.
This bound asymptotically improves on the classical
$\Theta\!\left(\frac{N}{B}\left[1+
\frac{\log(N/M)}{\log(M/B)}\right]\right)$
comparison-sorting bound in certain parameter regimes.
Under the external-memory model specified in this paper, we have
formally verified the algorithm's correctness and I/O complexity
bound in Lean~4. Moreover, our sorting reduction is independent of the particular construction used to interchange address fields. Substituting the address-field interchange construction from recent work on multitape integer multiplication into this reduction extends the sorting bound above to every fixed $0<\varepsilon<1/2$ (with a constant factor depending on $\varepsilon$).

\clearpage
\pagenumbering{arabic}
\section{Introduction}
\label{sec:introduction}

A multiple-unicast network carries independent messages between specified source--destination pairs. In 2004, Li and Li~\cite{LiLi2004} and, independently, Harvey, Kleinberg, and Lehman~\cite{HarveyKleinbergLehman2004} proposed the multiple-unicast conjecture: in undirected multiple-unicast networks, network coding, which allows nodes to mix messages, and fractional routing, which forwards messages without mixing them, achieve the same message rates. The conjecture is directly connected to several complexity lower bounds in theoretical computer science, including conditional lower bounds for external-memory integer sorting~\cite{FHLS2019} and for integer multiplication in constant-degree Boolean circuits~\cite{AFKL2019}.

In external-memory sorting, the $N$ input integers are stored on disk, with one integer per machine word. Memory holds $M$ words, and each transfer between disk and memory moves a block of $B$ consecutive words. The input/output (I/O) complexity is the total number of block transfers. The optimal comparison-sorting bound is~\cite{AV1988}
\begin{equation}
\Theta\left(\frac{N}{B}\left[1+\frac{\log(N/M)}{\log(M/B)}\right]\right).
\label{eq:intro-comparison}
\end{equation}
In 1988, Aggarwal and Vitter asked whether allowing algorithms to exploit the bit representations of integer keys could improve on the comparison bound in~\eqref{eq:intro-comparison}~\cite{AV1988}. For this open problem, Farhadi, Hajiaghayi, Larsen, and Shi~\cite{FHLS2019} proved that, if the multiple-unicast conjecture were true, then for $B\ge\log N$, external-memory integer sorting could not asymptotically improve on~\eqref{eq:intro-comparison}, even with arbitrary bit operations.

In 2026, Zhang, Li, and Li~\cite{ZLL2026} and Braverman and He~\cite{BH2026} constructed counterexamples in undirected multiple-unicast networks, showing that coding can achieve rates strictly higher than fractional routing and thereby refuting the conjecture. These results establish a communication advantage, but invalidating the premise of a conditional complexity lower bound does not itself produce a faster sorting algorithm.

This paper studies how to turn a communication advantage in undirected networks into fewer block transfers in external memory. Our construction of a counterexample to the multiple-unicast conjecture~\cite{LiLi2004,HarveyKleinbergLehman2004} yields a deterministic external-memory integer-sorting algorithm that reduces block transfers and answers the question of Aggarwal and Vitter~\cite{AV1988} in the affirmative. The counterexample builds on the finite-network construction underlying OpenAI's integer-multiplication algorithm~\cite{OpenAI2026} and subsequent circuit refinements~\cite{E8Base555,E8Word589}, and we use it to accelerate external-memory sorting. Earlier work used the connection between network coding and external-memory algorithms in the opposite direction. Assuming the multiple-unicast conjecture, it established a lower bound for matrix transposition with a fixed access pattern~\cite{Adler2006} and showed that, for $B \ge \log N$, even arbitrary bit operations cannot asymptotically improve on the comparison-sorting I/O bound for integer sorting~\cite{FHLS2019}. Our construction shows that coding can instead be used to achieve an asymptotic improvement over comparison sorting.

The connection between the two problems lies in how data is moved. Integer sorting relies on repeatedly rearranging records on disk according to groups of bits in their keys or array indices, a basic form of which is swapping two equal-width groups of digits in an array index~\cite{OpenAI2026}. When the dataset exceeds internal memory, standard divide-and-conquer methods divide a digit swap into smaller recursive subproblems that process the data streams separately without mixing different records, much as fractional routing forwards messages without combining them. Our counterexample provides a network in which combining messages achieves a higher communication rate than routing. Using the coding scheme in this network, the algorithm instead combines multiple data streams into encoded streams during sequential disk scans, applies the recursive rearrangements to the encoded streams, and then decodes them to recover each original stream in its target order. By exploiting this coding advantage, the algorithm reduces the total amount of data processed by recursive rearrangements compared with rearranging the streams separately.

Each level of this recursion shrinks the subproblems until their rearrangements can be completed by sequential scans with buffers that fit in the $M$-word memory. We combine this fast rearrangement routine with Columnsort and radix sorting to obtain a complete external-memory sorting algorithm for the full parameter range $2B \le M \le N$. More broadly, this approach may also help solve other problems where data movement is a bottleneck.

Specifically, for every $2B\le M\le N$, our algorithm deterministically sorts $N$ integers of $\Theta(\log N)$ bits in
\[
O\left(\frac{N}{B}\left[1+\frac{\log N}{(\log(M/B))^{1-\varepsilon}(\log M)^\varepsilon}\right]\right)
\]
I/Os, where the absolute constant can be taken as $\varepsilon=1.5\times10^{-5}>0$. Under the model conventions in Section~\ref{sec:results}, the algorithm is correct for every input, including repeated values, and uses at most $M$ words of memory and polynomial disk space. In particular, when $M=2B$ and $N=B^2$, comparison sorting requires $\Theta(B\log B)$ I/Os, whereas our algorithm uses $O(B(\log B)^{1-\varepsilon})$.

\paragraph{Formal verification.}
Using Lean~4~\cite{Lean2021}, we have formally verified the sorting algorithm's correctness and I/O complexity bound under the model of Section~\ref{sec:results}.\footnote{Public release repository: \url{https://github.com/rtyu-1/coded-sorting-lean}.} The verification covers $\varepsilon=1.5\times10^{-5}$ and the full parameter range of Theorem~\ref{thm:sorting}. Appendix~\ref{app:formalization} supplies the execution-control convention.

\paragraph{Improved sorting exponent.}
Our sorting reduction in Sections~7 and~8 invokes address-field interchange as a black box, independently of the specific circuit constructed in Sections~5 and~6. Building on this modularity, Appendix~\ref{sec:improved-interchange} adapts the address-field interchange construction from OpenAI's subsequent multitape integer multiplication work~\cite{OpenAITapes2026} to external memory and substitutes it into the same reduction, showing that the sorting bound above holds for every fixed $0<\varepsilon<1/2$ (the multiplicative constant may depend on $\varepsilon$ and is independent of $N$, $M$, and $B$).

Section~\ref{sec:related} reviews related work. Sections~\ref{sec:results}--\ref{sec:overview} give the model, main results, and proof overview. Section~\ref{sec:circuit-framework} constructs the coding counterexample, and Section~\ref{sec:interchange} applies it to address-field interchange. Sections~\ref{sec:sorting}--\ref{sec:global-sorting} construct the sorting algorithm and complete the main proof. The appendices supply memory simulation, fixed-parameter calculations, interchange for general field widths and prefix counts, the execution-control convention, and the improved sorting exponent via~\cite{OpenAITapes2026}.

\section{Related Work}
\label{sec:related}

\paragraph{External-memory sorting.}
Aggarwal and Vitter~\cite{AV1988} proved the comparison-sorting I/O lower bound in~\eqref{eq:intro-comparison} and showed that multiway merge sort and distribution sort attain it. They also proved that, even without charging for comparisons, the same lower bound holds when $B\ge\log N$ if each record is moved as an indivisible object. Whether splitting records into bits and encoding them could beat this bound remained unresolved. Subsequent external-memory sorting research mainly improved other aspects~\cite{Vitter2001}: Frigo, Leiserson, Prokop, and Ramachandran~\cite{FLPR1999} introduced cache-oblivious sorting under the assumption $M=\Omega(B^2)$, without requiring knowledge of $M$ or $B$; Arge and Thorup~\cite{ArgeThorup2015} used internal-memory integer-sorting techniques to reduce internal computation time. Neither improves on the I/O bound in~\eqref{eq:intro-comparison}.

Encoding integer representations has yielded I/O improvements for external-memory problems other than sorting. Iacono and P\u{a}tra\c{s}cu~\cite{IaconoPatrascu2012} used hashing to beat comparison-based dictionary-update bounds and proved a matching lower bound in the external-memory cell-probe model, which charges for block accesses but not computation. Jiang and Larsen~\cite{JiangLarsen2019} constructed faster external-memory priority queues supporting DecreaseKey, which lowers the priority of an element already in the queue. For general integer sorting, Farhadi, Hajiaghayi, Larsen, and Shi~\cite{FHLS2019} proved an I/O lower bound conditional on the multiple-unicast conjecture, even allowing arbitrary bit operations. Since the conjecture has been refuted, this conditional result does not establish an unconditional sorting lower bound.

Our sorting construction combines permutation-based sorting with coded data rearrangement. Leighton's Columnsort~\cite{Leighton1985} decomposes sorting into alternating column sorts and fixed intercolumn permutations. Chaudhry, Cormen, and Wisniewski~\cite{CCW2001} subsequently gave an external-memory implementation, and Chaudhry and Cormen~\cite{ChaudhryCormen2006} further developed implementations for distributed-memory clusters. This framework reduces all interaction between columns to fixed address rearrangements independent of key values. Cormen, Sundquist, and Wisniewski~\cite{CSW1998} analyzed the I/O cost of a class of permutations on parallel disks that includes swaps of binary address fields. In a different line of work, OpenAI~\cite{OpenAI2026} used a finite linear circuit to reduce the total data volume processed by recursive subproblems, accelerating address-field interchange and the resulting fixed permutations on a multitape Turing machine. We extend this coded interchange mechanism to external-memory block transfers, accelerate the fixed permutations in Columnsort, and use radix sorting to handle arbitrary input sizes, improving on the comparison-sorting I/O bound, for example when $M=2B$ and $N=B^2$. Using a refined address-field interchange construction, OpenAI~\cite{OpenAITapes2026} obtained time $n(\log n)^{1/2+o(1)}$ for multiplying two $n$-bit integers on a multitape Turing machine. Appendix~\ref{sec:improved-interchange} adapts the interchange construction to deterministic external-memory sorting (Corollary~\ref{cor:near-half-sorting}).

\paragraph{The multiple-unicast conjecture.}
Li and Li~\cite{LiLi2004} and, independently, Harvey, Kleinberg, and Lehman~\cite{HarveyKleinbergLehman2004} conjectured that network coding and fractional routing have the same achievable rates for independent multiple-unicast sessions in undirected networks. One class of positive results follows from cut sufficiency for supply and demand graphs. Every coding-achievable rate satisfies the cut condition; hence, whenever the supply and demand graphs make that condition sufficient for a feasible multicommodity flow, coding equals routing. Examples include the cut-sufficient pairs with series-parallel supply graphs characterized by Chakrabarti, Fleischer, and Weibel~\cite{CFW2012}. When cuts do not characterize routing feasibility, subsequent work has proved restricted cases of the conjecture along three directions: using message independence and entropy inequalities to obtain stronger information-theoretic bounds~\cite{Jain2006,Harvey2006,QureshiThakor2022}; simplifying networks and sessions through graph reductions~\cite{Yin2018} or session dominance and interaction frameworks~\cite{LiuSession2026,LiuInteraction2026}; and developing geometric and metric methods that equate coding and routing costs through transmission costs and metric embeddings~\cite{Xiahou2014,LiuLi2025,LiuBeyondCut2026}.

The conjecture is also directly connected to computational lower bounds. Adler, Harvey, Jain, Kleinberg, and Lehman~\cite{Adler2006} first used it to prove conditional lower bounds in several models whose access patterns are independent of the input. Farhadi, Hajiaghayi, Larsen, and Shi~\cite{FHLS2019} removed this restriction for external-memory integer sorting: assuming the conjecture, any randomized sorting algorithm with error probability at most $1/3$ and key length $w=\Omega(\log N)$ requires at least
\[
\Omega\left(\min\left\{N,\frac{N}{B}\log_{2M/B}\frac{N}{B}\right\}\right)
\]
expected I/Os. When $B\ge\log N$, this lower bound matches the comparison-sorting bound up to constant factors. Afshani, Freksen, Kamma, and Larsen~\cite{AFKL2019} extended this connection to integer multiplication, obtaining conditional lower bounds for constant-degree Boolean circuits.

The finite binary linear circuit in OpenAI's integer-multiplication construction~\cite[Section~3]{OpenAI2026} also implies a counterexample to the multiple-unicast conjecture: applying the routing argument of Section~\ref{sec:network} to its rational matrix labels bounds fractional routing strictly below an achievable coding rate.

Recently, Zhang, Li, and Li~\cite{ZLL2026} and Braverman and He~\cite{BH2026} gave counterexamples to the multiple-unicast conjecture, proving that coding can achieve strictly higher throughput than fractional routing. The former start from examples where coding completes transmission sooner than routing and construct a reversible communication process; the latter use finite geometry and polynomial evaluation to construct a coding scheme. These results establish a communication advantage, but refuting the conjecture does not itself produce a faster sorting algorithm. We further study how this advantage can reduce data movement in computation: our finite linear circuit both gives a strict separation of communication rates in undirected multiple-unicast networks and supports recursive external-memory address permutations, turning a coding advantage into an improved I/O bound for external-memory integer sorting.

\section{Main Results}
\label{sec:results}

We write $\log$ for the base-$2$ logarithm and $\ln$ for the natural logarithm. For a positive integer $t$, let $[t]=\{0,\ldots,t-1\}$.

\subsection{The external-memory model and sorting bound}
Let $N,M,B,w$ be positive integers with $N\ge2$. The input consists of $N$ integers of $w$ bits stored on disk, where $w=\Theta(\log N)$. A storage unit of $w$ bits is called a \emph{machine word}, and each input integer occupies one word. The algorithm must write the integers back to disk in nondecreasing order.

Memory holds $M$ words, each disk block contains $B$ consecutive words, and $2B\le M\le N$. A \emph{block transfer} reads one disk block into memory or writes one block from memory to disk. The algorithm may access any disk block, and its complexity is the total number of block transfers.

Disk data must be read into memory before they can participate in computation; computation in memory is not charged as I/O. The algorithm may perform arbitrary bit operations on integers, including combining bits from different integers into intermediate results. Data buffers and mutable workspace count toward the $M$-word memory bound.

The program is fixed from the size parameters, independently of the input values, and its code is excluded from data space. Appendix~\ref{app:formalization} gives the detailed model convention.

The input, output, and intermediate data on disk together occupy at most $O(N^k)$ words for a constant $k$ independent of $N,M,B$. We call this \emph{polynomial disk space}. Disk addresses therefore occupy $O(\log N)$ bits, or a constant number of machine words.

\begin{theorem}[External-memory integer sorting]
\label{thm:sorting}
Under these conventions, there is a constant $\varepsilon>0$ such that, for every $2B\le M\le N$, $N$ integers of $\Theta(\log N)$ bits can be sorted deterministically in
\begin{equation}
 O\!\left(\frac NB\left[1+
 \frac{\log N}{(\log(M/B))^{1-\varepsilon}(\log M)^\varepsilon}
 \right]\right)
 \label{eq:simple-sorting}
\end{equation}
block transfers. The algorithm uses at most $M$ words of memory and polynomial disk space, and is correct on every input.
\end{theorem}

Here $\varepsilon$ is independent of $N,M,B$; we take $\varepsilon=1.5\times10^{-5}$. When the input fits in memory, reading, sorting, and writing it takes the optimal $O(N/B)$ block transfers. In general, $N/B$ in~\eqref{eq:simple-sorting} is the cost of scanning the input. In the classical comparison bound~\eqref{eq:intro-comparison}, the logarithmic factor beyond scanning is $\log(N/M)/\log(M/B)$, whose denominator is governed entirely by the number of blocks $M/B$ that fit in memory. When $\log(N/M)=\Theta(\log N)$, for example when $N\ge M^2$, the improvement in Theorem~\ref{thm:sorting} can be viewed as replacing $(\log(M/B))^\varepsilon$ in the denominator by $(\log M)^\varepsilon$. This reduces block transfers when memory holds few blocks. Table~\ref{tab:examples} gives two examples with $M=2B$, each with an asymptotic improvement factor of $(\log B)^\varepsilon$ over comparison sorting.

\begin{table}[htbp]
\centering
\normalsize
\renewcommand{\arraystretch}{1.4}
\begin{tabular}{@{}lll@{}}
\toprule
Input size ($M=2B$) & Classical comparison sorting & This work\\
\midrule
$N=B^2$ & $\Theta(B\log B)$ & $O(B(\log B)^{1-\varepsilon})$\\
$N=B^3$ & $\Theta(B^2\log B)$ & $O(B^2(\log B)^{1-\varepsilon})$\\
\bottomrule
\end{tabular}
\caption{Improvements for two input sizes, with $B\to\infty$ in both cases.}
\label{tab:examples}
\end{table}

\subsection{The communication model and coding counterexample}

Let $G=(V,E)$ be a finite undirected multiple-unicast network. An edge $e$ has capacity $c(e)$, shared by the two directions: their transmission rates must sum to at most $c(e)$. The network has $W$ unicast sessions. Session $i$ has source $s_i$, sink $t_i\ne s_i$, and required rate $r_i$. Source $s_i$ holds a message $X_i$ that sink $t_i$ must recover. The source messages are mutually independent, and the rate vector is $\mathbf r=(r_1,\ldots,r_W)$.

Under \emph{fractional routing}, each session may use several source--sink paths. Assign nonnegative flow to the paths of session $i$ so that their sum is $r_i$. For every edge, the total flow on all paths using it must not exceed its capacity; this sum includes both directions and all sessions.

\emph{Network coding} allows nodes to transmit functions of their source messages and previously received information. Each sink must recover its required message for every source input, giving a zero-error requirement. For fixed network, capacities, and source--sink pairs, the simultaneously achievable rate vectors and their limits form the corresponding \emph{capacity region}. The multiple-unicast conjecture asserts that the coding and fractional routing capacity regions coincide. A rate vector achievable by coding but not by routing therefore refutes it.

All edges in our counterexample have unit capacity. We use a finite-round binary linear protocol: every transmitted bit is an exclusive-or (XOR), or sum modulo two, of bits already available at the transmitting node, and transmission times and directions are fixed in advance. In each round, each edge carries at most one bit in one direction, computed only from initial messages or information received in earlier rounds. Local computation and storage are unrestricted. If a protocol delivers $k$ bits from each source to its sink in $T$ rounds, its \emph{common rate} is $k/T$. Let $R_{\mathrm{route}}$ be the maximum common rate under fractional routing, and let $R_{\mathrm{code}}$ be the supremum of common rates achievable by these binary linear protocols.

The finite construction used in this paper also yields the following communication result.

\begin{theorem}[A strict coding advantage in undirected multiple-unicast networks]
\label{thm:network}
There is a finite undirected multiple-unicast network with no loops or parallel edges, unit edge capacities, and independent source messages such that
\begin{equation}
 R_{\mathrm{route}}
 \le 1-\frac{33}{338432}
 <1=R_{\mathrm{code}}.
 \label{eq:network-separation}
\end{equation}
\end{theorem}

\Needspace{14\baselineskip}
\section{Proof Overview}
\label{sec:overview}

External-memory sorting must move data on disk to new positions. Our central approach uses a finite XOR circuit to reduce the number of block transfers needed for this movement. The same circuit has two roles. In an undirected multiple-unicast network, it delivers independent messages to their respective sinks at a rate higher than fractional routing, giving a counterexample to the multiple-unicast conjecture. In an external-memory algorithm, mixing and recovering bits from different positions with this circuit breaks a large movement task into smaller tasks whose total data volume is less than that required by directly interchanging corresponding short fields. We first introduce the rearrangements needed for sorting, then the circuit and coding counterexample, and finally their use in address interchange and external-memory sorting.

\subsection{Rearrangement in sorting}
\label{sec:overview-rearrangement}

Arrange the input integers in a two-dimensional table, storing each column contiguously on disk. Sorting each column separately does not remove inversions between different columns. Columnsort alternates column sorts with fixed intercolumn rearrangements. These rearrangements redistribute the data among new columns, bringing integers from different original columns together for comparison. When the row and column counts satisfy Columnsort's conditions, executing the prescribed number of column sorts and rearrangements in order makes the entire table sorted.

A rearrangement changes only the storage positions of elements, preserving their values. Each element's new position is determined entirely by its original row and column indices, independently of its value. Figure~\ref{fig:method-overview} gives a matrix-transpose example. These fixed movements determined by indices are precisely what our coding circuit accelerates. Intermediate steps may XOR bits from different integers, but at the end of the rearrangement each target position contains the input integer that belongs there. The intercolumn rearrangements required by Columnsort also have a regular structure: each decomposes into a constant number of interchanges of two address fields, together with sequential scans that read or write the array in storage order.

\begin{figure}[!htbp]
\centering
\begin{tikzpicture}[font=\normalsize]
\node[anchor=west,font=\normalsize\bfseries] at (0,0) {(a) Fixed rearrangement};
\node[datacell,minimum width=0.56cm,minimum height=0.56cm,fill=white] (a00) at (2.720,-0.550) {8};
\node[datacell,minimum width=0.56cm,minimum height=0.56cm,fill=lightblue] (a01) at (3.280,-0.550) {3};
\node[datacell,minimum width=0.56cm,minimum height=0.56cm,fill=warmfill] (a10) at (2.720,-1.110) {6};
\node[datacell,minimum width=0.56cm,minimum height=0.56cm,fill=white] (a11) at (3.280,-1.110) {1};
\node[datacell,minimum width=0.56cm,minimum height=0.56cm,fill=white] (b00) at (8.720,-0.550) {8};
\node[datacell,minimum width=0.56cm,minimum height=0.56cm,fill=warmfill] (b01) at (9.280,-0.550) {6};
\node[datacell,minimum width=0.56cm,minimum height=0.56cm,fill=lightblue] (b10) at (8.720,-1.110) {3};
\node[datacell,minimum width=0.56cm,minimum height=0.56cm,fill=white] (b11) at (9.280,-1.110) {1};
\draw[figarrow] (4,-.83) -- node[above,smalltag] {Transpose} (8,-.83);
\node[anchor=west,font=\normalsize\bfseries] at (0,-1.9) {(b) Column sorts and rearrangement};
\node[datacell,minimum width=0.49cm,minimum height=0.49cm,fill=lightblue] (m000) at (0.610,-2.530) {9};
\node[datacell,minimum width=0.49cm,minimum height=0.49cm,fill=warmfill] (m001) at (1.100,-2.530) {2};
\node[datacell,minimum width=0.49cm,minimum height=0.49cm,fill=greenfill] (m002) at (1.590,-2.530) {7};
\node[datacell,minimum width=0.49cm,minimum height=0.49cm,fill=lightblue] (m010) at (0.610,-3.020) {4};
\node[datacell,minimum width=0.49cm,minimum height=0.49cm,fill=warmfill] (m011) at (1.100,-3.020) {8};
\node[datacell,minimum width=0.49cm,minimum height=0.49cm,fill=greenfill] (m012) at (1.590,-3.020) {1};
\node[datacell,minimum width=0.49cm,minimum height=0.49cm,fill=lightblue] (m020) at (0.610,-3.510) {6};
\node[datacell,minimum width=0.49cm,minimum height=0.49cm,fill=warmfill] (m021) at (1.100,-3.510) {3};
\node[datacell,minimum width=0.49cm,minimum height=0.49cm,fill=greenfill] (m022) at (1.590,-3.510) {5};
\draw[figarrow] (2.00,-3.02) -- node[above,smalltag] {Sort} (3.80,-3.02);
\node[datacell,minimum width=0.49cm,minimum height=0.49cm,fill=lightblue] (m100) at (4.210,-2.530) {4};
\node[datacell,minimum width=0.49cm,minimum height=0.49cm,fill=warmfill] (m101) at (4.700,-2.530) {2};
\node[datacell,minimum width=0.49cm,minimum height=0.49cm,fill=greenfill] (m102) at (5.190,-2.530) {1};
\node[datacell,minimum width=0.49cm,minimum height=0.49cm,fill=lightblue] (m110) at (4.210,-3.020) {6};
\node[datacell,minimum width=0.49cm,minimum height=0.49cm,fill=warmfill] (m111) at (4.700,-3.020) {3};
\node[datacell,minimum width=0.49cm,minimum height=0.49cm,fill=greenfill] (m112) at (5.190,-3.020) {5};
\node[datacell,minimum width=0.49cm,minimum height=0.49cm,fill=lightblue] (m120) at (4.210,-3.510) {9};
\node[datacell,minimum width=0.49cm,minimum height=0.49cm,fill=warmfill] (m121) at (4.700,-3.510) {8};
\node[datacell,minimum width=0.49cm,minimum height=0.49cm,fill=greenfill] (m122) at (5.190,-3.510) {7};
\draw[figarrow] (5.60,-3.02) -- node[above,smalltag] {Transpose} (7.40,-3.02);
\node[datacell,minimum width=0.49cm,minimum height=0.49cm,fill=lightblue] (m200) at (7.810,-2.530) {4};
\node[datacell,minimum width=0.49cm,minimum height=0.49cm,fill=lightblue] (m201) at (8.300,-2.530) {6};
\node[datacell,minimum width=0.49cm,minimum height=0.49cm,fill=lightblue] (m202) at (8.790,-2.530) {9};
\node[datacell,minimum width=0.49cm,minimum height=0.49cm,fill=warmfill] (m210) at (7.810,-3.020) {2};
\node[datacell,minimum width=0.49cm,minimum height=0.49cm,fill=warmfill] (m211) at (8.300,-3.020) {3};
\node[datacell,minimum width=0.49cm,minimum height=0.49cm,fill=warmfill] (m212) at (8.790,-3.020) {8};
\node[datacell,minimum width=0.49cm,minimum height=0.49cm,fill=greenfill] (m220) at (7.810,-3.510) {1};
\node[datacell,minimum width=0.49cm,minimum height=0.49cm,fill=greenfill] (m221) at (8.300,-3.510) {5};
\node[datacell,minimum width=0.49cm,minimum height=0.49cm,fill=greenfill] (m222) at (8.790,-3.510) {7};
\draw[figarrow] (9.20,-3.02) -- node[above,smalltag] {Sort} (11.00,-3.02);
\node[datacell,minimum width=0.49cm,minimum height=0.49cm,fill=greenfill] (m300) at (11.410,-2.530) {1};
\node[datacell,minimum width=0.49cm,minimum height=0.49cm,fill=warmfill] (m301) at (11.900,-2.530) {3};
\node[datacell,minimum width=0.49cm,minimum height=0.49cm,fill=greenfill] (m302) at (12.390,-2.530) {7};
\node[datacell,minimum width=0.49cm,minimum height=0.49cm,fill=warmfill] (m310) at (11.410,-3.020) {2};
\node[datacell,minimum width=0.49cm,minimum height=0.49cm,fill=greenfill] (m311) at (11.900,-3.020) {5};
\node[datacell,minimum width=0.49cm,minimum height=0.49cm,fill=warmfill] (m312) at (12.390,-3.020) {8};
\node[datacell,minimum width=0.49cm,minimum height=0.49cm,fill=lightblue] (m320) at (11.410,-3.510) {4};
\node[datacell,minimum width=0.49cm,minimum height=0.49cm,fill=lightblue] (m321) at (11.900,-3.510) {6};
\node[datacell,minimum width=0.49cm,minimum height=0.49cm,fill=lightblue] (m322) at (12.390,-3.510) {9};
\end{tikzpicture}
\caption{(a) A fixed position interchange preserves values. (b) Colors identify original columns; rearrangement brings different colors into the same column. The figure illustrates the operations only; Section~\ref{sec:sorting} gives the full Columnsort procedure.}
\label{fig:method-overview}
\end{figure}

\Needspace{9\baselineskip}
\subsection{Circuits and the coding counterexample}
\label{sec:overview-circuit}

A coding counterexample requires both a coding scheme that recovers every message and a proof that fractional routing on the same network cannot achieve that rate. A finite circuit specifies a predetermined sequence of XOR operations. Intermediate steps may mix bits of different messages, but every original bit is recovered at its designated output. Turning circuit wires into network edges gives a coding scheme; assigning matrix labels to nodes bounds the routing rate.

Figure~\ref{fig:circuit-network}(a) shows a circuit with three registers $A,B,C$ and six gates. Each register stores one bit, initially $x,y,z$, respectively. Each horizontal line represents a register, and operations proceed from left to right. A black dot and an $\oplus$ on the same vertical line form a gate: the black-dot value is unchanged, and the $\oplus$ value is XORed with it. For example, the first gate $g_1$ changes register $C$ from $z$ to $z\oplus x$. For every input $(x,y,z)$, the six gates output $(y,x,z)$. Registers $A$ and $B$ exchange their initial values; we call two registers that exchange initial values a pair of \emph{data registers}. Register $C$ assists this exchange and is restored to $z$ at the end, so we call it an \emph{auxiliary register}. Because its initial value is arbitrary and is restored, an auxiliary register can also carry an independent message and need not start at zero.

% A six-gate binary-linear example: (x,y,z) -> (y,x,z).
% Gate list (control,target): (A,C),(C,B),(A,C),(C,B),(B,A),(A,B).
% Colors identify register wires, not independently routed messages.
\begin{figure}[!htbp]
\centering
\begin{tikzpicture}[font=\normalsize,
 endpoint/.style={draw=inkblue!65,rounded corners=1pt,minimum width=.65cm,minimum height=.48cm,inner sep=2pt},
 gatepoint/.style={circle,fill=inkblue,inner sep=2.5pt},
 terminal/.style={circle,draw=inkblue,fill=white,inner sep=2.2pt},
 bitlabel/.style={font=\normalsize,inner sep=1.5pt,fill=white}]

\node[anchor=west,font=\normalsize\bfseries] at (-1.15,1.10) {(a) Three registers, six gates};
\draw[figarrow,draw=gray!65] (7.5,1.10) -- node[above,smalltag] {Execution order} (10.85,1.10);

% Each horizontal line is one fixed register.
\foreach \y/\reg/\val/\out/\col/\fillcol in {0/A/x/y/inkblue/lightblue,-1/B/y/x/warmink/warmfill,-2/C/z/z/greenink/greenfill} {
 \node[anchor=east,text=\col] at (-.78,\y) {$\reg$};
 \node[endpoint,fill=\fillcol] (in\reg) at (-.22,\y) {$\val$};
 \node[endpoint,fill=\fillcol] (out\reg) at (11.1,\y) {$\out$};
 \draw[draw=\col,semithick] (in\reg.east) -- (out\reg.west);
}

% A control dot leaves its own register unchanged; the target takes XOR.
\foreach \xx/\cy/\ty/\j in {1/0/-2/1,2.8/-2/-1/2,4.6/0/-2/3,6.4/-2/-1/4,8.2/-1/0/5,10/0/-1/6} {
 \draw[gray!80,semithick] (\xx,\cy) -- (\xx,\ty);
 \fill[inkblue] (\xx,\cy) circle (1.8pt);
 \draw[draw=inkblue,fill=white,semithick] (\xx,\ty) circle (4pt);
 \draw[draw=inkblue,semithick] (\xx-.105,\ty) -- (\xx+.105,\ty);
 \draw[draw=inkblue,semithick] (\xx,\ty-.105) -- (\xx,\ty+.105);
 \node[smalltag] at (\xx,.52) {$g_{\j}$};
}

% Selected intermediate values make cancellation visible.
\node[bitlabel,text=greenink] at (1.9,-2.24) {$z\oplus x$};
\node[bitlabel,text=warmink] at (3.7,-.73) {$y\oplus z\oplus x$};
\node[bitlabel,text=greenink] at (5.5,-2.24) {$z$};
\node[bitlabel,text=warmink] at (7.3,-.73) {$y\oplus x$};
\node[bitlabel,text=inkblue] at (9.1,.23) {$y$};

\draw[figarrow] (5.45,-2.65) -- (5.45,-3.25);
\node[smalltag,anchor=west] at (5.75,-2.95) {Gate $\to$ node, wire segment $\to$ edge};
\node[anchor=west,font=\normalsize\bfseries] at (-1.15,-3.65) {(b) The corresponding undirected multiple-unicast network};

% Distinct source and destination vertices, including the restored auxiliary bit.
\node[terminal,label=left:{$s_x$}] (sx) at (-.22,-4.35) {};
\node[terminal,label=left:{$s_z$}] (sz) at (-.22,-5.20) {};
\node[terminal,label=left:{$s_y$}] (sy) at (-.22,-6.10) {};
\node[terminal,label=right:{$t_y$}] (ty) at (11.1,-4.35) {};
\node[terminal,label=right:{$t_x$}] (tx) at (11.1,-5.20) {};
\node[terminal,label=right:{$t_z$}] (tz) at (11.1,-6.10) {};
\foreach \xx/\j in {1/1,2.8/2,4.6/3,6.4/4,8.2/5,10/6} {
 \coordinate (g\j) at (\xx,-5.20);
}

% Register A: sx -> g1 -> g3 -> g5 -> g6 -> ty.
\draw[inkblue,semithick] (sx) -- (g1);
\draw[inkblue,semithick] (g1) to[out=65,in=115] (g3);
\draw[inkblue,semithick] (g3) to[out=65,in=115] (g5);
\draw[inkblue,semithick] (g5) to[out=40,in=140] (g6);
\draw[inkblue,semithick] (g6) -- (ty);

% Register B: sy -> g2 -> g4 -> g5 -> g6 -> tx.
\draw[warmink,semithick] (sy) to[out=0,in=-130] (g2);
\draw[warmink,semithick] (g2) to[out=-60,in=-120] (g4);
\draw[warmink,semithick] (g4) -- (g5);
\draw[warmink,semithick] (g5) to[out=-40,in=-140] (g6);
\draw[warmink,semithick] (g6) -- (tx);

% Register C: sz -> g1 -> g2 -> g3 -> g4 -> tz.
\draw[greenink,semithick] (sz) -- (g1) -- (g2) -- (g3) -- (g4);
\draw[greenink,semithick] (g4) to[out=-50,in=180] (tz);

% Gate vertices are drawn last so crossings and junctions are unambiguous.
\foreach \j in {1,...,6} {
 \node[gatepoint] at (g\j) {};
 \node[smalltag,anchor=south] at ($(g\j)+(0,.09)$) {$g_{\j}$};
}
\end{tikzpicture}
\caption{A three-register example and its network. Colors distinguish register wires. Node $g_j$ below represents gate $j$ above; an edge represents a wire segment between consecutive gates or endpoints.}
\label{fig:circuit-network}
\end{figure}
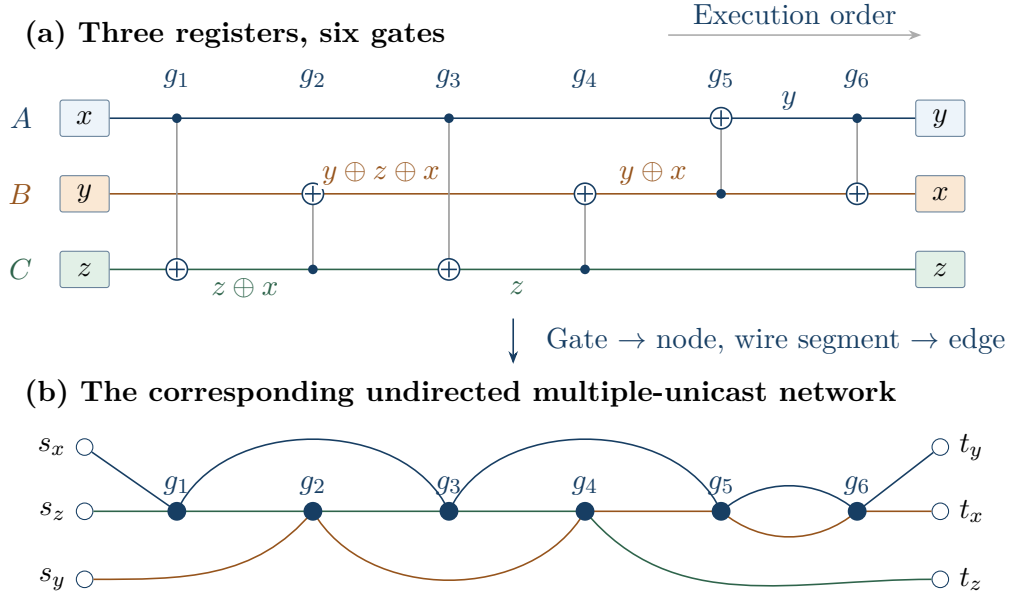

Figure~\ref{fig:circuit-network}(b) turns each gate $g_j$ into a node and each segment between successive gates or endpoints on the same register into a unit-capacity undirected edge. The three sessions deliver $x,y,z$ to destinations bearing the same letters. Nodes receive, compute, and transmit in circuit order, recovering every message. Pipelining independent executions allows every edge to carry at most one bit per round. For a fixed startup delay $D$, every source can deliver $k$ bits in at most $k+D$ rounds. As $k$ grows, the coding rate $k/(k+D)$ approaches one.

To prove that fractional routing on the network constructed in Section~\ref{sec:finite-construction} cannot achieve rate one, assign an $m\times m$ rational matrix to each node, requiring the sink matrix minus the source matrix of each message to have full rank $m$. Define an edge's length as the rank of the difference between its endpoint matrices. These lengths serve only to bound routing; the actual edge capacities remain one. Along any source--sink path, the successive matrix differences sum to the full-rank difference between its endpoints. Traversing an edge in reverse changes only the sign of its matrix difference, not its rank. Rank subadditivity therefore makes every source--sink path at least $m$ long, including paths with edges traversed against the circuit's time direction. Let there be $W$ independent messages and let $S$ be the sum of all edge lengths. Routing all messages at common rate $r$ consumes total edge length at least $rWm$, whereas the unit edge capacities supply only $S$. If the circuit and labels satisfy
\begin{equation}
 S<Wm,
 \label{eq:overview-deficit}
\end{equation}
then $r\le S/(Wm)<1$, so the coding rate eventually exceeds the routing upper bound.

\paragraph{Where the saving comes from.}
Like the example in Figure~\ref{fig:circuit-network}(a), the circuit in Section~\ref{sec:finite-construction} consists of pairs of data registers that exchange initial values and auxiliary registers that recover their initial values. If every source--sink path has length at least $m$, how can the sum $S$ of all edge lengths be smaller than $Wm$? In a network like Figure~\ref{fig:circuit-network}(b), XOR gates are nodes; its edges are the horizontal register segments. The two messages of each data pair exchange output positions, so a message's source and sink need not lie on the same horizontal line.

In the construction, node matrices are orthogonal projections onto subspaces of an $m$-dimensional space. When the subspace labels on a horizontal register line grow by inclusion, the edge lengths along that line sum to the difference between its endpoint dimensions. The total $S$ has three contributions. First, each auxiliary-register line grows from dimension zero to dimension $m$, contributing exactly $m$. Second, we arrange for each of the two lines of a data pair to grow by only $m-1$ dimensions, contributing $2m-2$ in total. Exchanging their output positions still leaves each source--sink matrix difference of rank $m$, so each pair saves two units relative to $2m$. Third, a few temporary wires add length: to transmit certain intermediate results, the circuit copies an auxiliary register's current value onto a separate wire for subsequent XOR gates and terminates the wire after use.

To obtain $S<Wm$, the savings on data-register lines must exceed the additional temporary-wire length. Section~\ref{sec:finite-coding} combines local circuits along three coordinate directions, using three rounds to exchange the initial values of each data pair. Each local invocation handles $111$ data pairs and adds temporary-wire length $63$. Over three rounds, the average additional length per data pair is $3\times63/111<2$, less than its saving of two, giving a strictly positive net saving. We also use the fact that auxiliary registers recover their initial values to share a pool among different parts of the construction while keeping their labels nested. Each auxiliary-register line still contributes only $m$ over all its uses, rather than $m$ per use. This reduces the number of auxiliary registers and increases the net saving as a fraction of $Wm$. Section~\ref{sec:finite-construction} gives the circuit and matrix labels.

\FloatBarrier
\subsection{Using coding to accelerate address interchange}
\label{sec:overview-interchange}

An array element's address is its position numbered from zero. A consecutive group of digits in its address representation is a \emph{field}. Interchanging two fields computes each element's new position by swapping the values of those fields and keeping all other digits fixed; the element's value is preserved. In Figure~\ref{fig:address-segments}, the element at position $28$ moves to position $42$ after the two colored fields are interchanged.

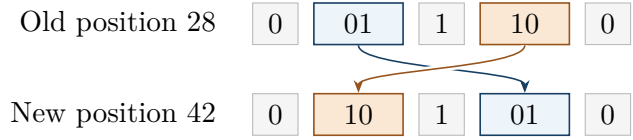
\begin{figure}[!htbp]
\centering
\begin{tikzpicture}[
 font=\normalsize,
 fixed/.style={draw=gray!55,fill=gray!8,minimum width=.60cm,minimum height=.55cm,inner sep=1pt},
 field/.style={minimum width=1.18cm,minimum height=.55cm,inner sep=1pt,semithick},
 bluefield/.style={field,draw=inkblue,fill=lightblue},
 orangefield/.style={field,draw=warmink,fill=warmfill}]
 \node[anchor=east] at (-.65,0) {Old position $28$};
 \node[anchor=east] at (-.65,-1.20) {New position $42$};
 \node[fixed] (p0) at (0,0) {$0$};
 \node[bluefield] (h0) at (1.10,0) {$01$};
 \node[fixed] (g0) at (2.20,0) {$1$};
 \node[orangefield] (d0) at (3.30,0) {$10$};
 \node[fixed] (z0) at (4.40,0) {$0$};
 \node[fixed] (p1) at (0,-1.20) {$0$};
 \node[orangefield] (d1) at (1.10,-1.20) {$10$};
 \node[fixed] (g1) at (2.20,-1.20) {$1$};
 \node[bluefield] (h1) at (3.30,-1.20) {$01$};
 \node[fixed] (z1) at (4.40,-1.20) {$0$};
 \draw[figarrow] (h0.south) .. controls +(0,-.30) and +(0,.30) .. (h1.north);
 \draw[figarrow,draw=warmink,preaction={draw=white,line width=2.5pt}]
   (d0.south) .. controls +(0,-.30) and +(0,.30) .. (d1.north);
\end{tikzpicture}
\caption{A binary address divided into five fields. The colored fields are interchanged; the gray fields stay in place.}
\label{fig:address-segments}
\end{figure}

Suppose the array contains $V$ bits, and the two equal-width fields to interchange are $H,D$. A direct divide-and-conquer approach splits each field into $m$ equal short fields and interchanges the $m$ corresponding pairs in sequence. This requires $m$ interchanges on the entire array, each with width $1/m$ of the original, for a total child data volume of $mV$. Our goal is to organize this recursion using the XOR circuit of Section~\ref{sec:overview-circuit}, with $W$ inputs and total edge length $S<Wm$, reducing the total child volume to $(S/W)V<mV$.

First use the address prefix preceding $H$ to divide the array equally among the circuit's $W$ inputs. Each input holds an ordered sequence of $V/W$ bits on disk, called a \emph{stream}. Interpret the node matrices as storage orders for these streams. Regard the $m$ short fields of $H,D$ as coordinates of two $m$-dimensional vectors. For an orthogonal projection matrix $P$ used by the circuit labels, define the following address transformation, called a \textbf{partial swap}:
\[
 (H,D)\longmapsto((I-P)H+PD,\;PH+(I-P)D),
\]
where $I$ is the identity matrix. It exchanges the components in the image of $P$ and preserves the components in its orthogonal complement. For $P=0$ the address is unchanged; for $P=I$ the complete fields $H,D$ are interchanged. These matrix operations determine the new addresses of the data; Section~\ref{sec:coding-recursion} specifies the arithmetic used. The data stored on disk remain bits, and the circuit gates still perform bitwise XOR.

Execute the circuit on the $W$ streams in gate order. At each XOR gate, all participating streams must be stored in the order specified by that gate's projection. Since the input streams have identical storage order, one sequential scan that XORs bits at matching disk positions implements the gate. When a stream travels from one gate to the next, convert its storage order to that required by the next projection. For an edge whose endpoint subspaces are nested, this conversion is exactly a partial swap: its projection is the larger subspace's projection minus the smaller one's, and its rank is the edge length. Temporary streams undergo the same conversions, whose costs are included in $S$. At the outputs, the XOR operations recover the original input bits, while the source and sink partial swaps together put each recovered stream into the fully interchanged address order $P=I$. Because data-register pairs exchange output positions, a final scan swaps the paired output streams back and combines them according to the original input assignment, completing the array's address-field interchange.

For an edge of length $r$, the corresponding partial swap exchanges coordinates in $r$ independent one-dimensional directions. Section~\ref{sec:interchange} proves that this partial swap can be implemented by $r$ short-field interchanges and a constant number of sequential disk scans. The fields interchanged on the current stream each have $1/m$ of the original width. Summing over all edges gives $S$ child interchanges, each on a stream of $V/W$ bits and fields $1/m$ as wide. Every child is again an interchange of two equal-width address fields, so recursion continues. Since $S<Wm$, there is a fixed $\varepsilon>0$ with $S/W<m^{1-\varepsilon}$, making the savings at successive recursion levels accumulate.

Stop when the field width reaches $\Theta(\log(M/B))$ bits. The buffers needed to group data by all pairs of short-field values then fit in memory, and a constant number of sequential scans completes the interchange. The entire stream need not fit in memory. Section~\ref{sec:interchange} and the appendices prove this recursion and handle general field widths and rounding in the stream division.

\FloatBarrier
\subsection{From rearrangement to sorting}

Section~\ref{sec:overview-interchange} accelerates address-field interchange. Such an interchange depends only on an element's current index, whereas its sorted destination depends on the values of the input records. We therefore decompose value-dependent sorting into local sorts and global rearrangements of fixed patterns.

For arrays of at most $M^2$ records, use Leighton's Columnsort. Arrange the array as a matrix with many more rows than columns. A fixed number of alternating rounds suffices: odd rounds sort each column independently, and even rounds apply predetermined global rearrangements, such as redistributing expanded columns into new row and column positions or shifting all elements by half a column. These even-round operations do not depend on input values. The intercolumn permutations rearrange the row and column indices within element addresses; the half-column shift and its padding and deletion use scans. Thus each global rearrangement can be implemented by a constant number of address-field interchanges from Section~\ref{sec:overview-interchange} and sequential scans.

If a column in an odd round still exceeds memory capacity, arrange it as a smaller matrix with many more rows than columns and recursively apply Columnsort. Once a subcolumn fits in memory, read it in and sort it there. With a constant recursion depth, summing the costs of the repeated global rearrangements introduces only a constant factor and preserves the exponent $1-\varepsilon$ in the complexity bound. Our matrix shape uses a number of columns of order the current size to the $1/4$ power, and a column length of order the size to the $3/4$ power. Each recursive level therefore reduces the size to a fractional power. Starting with at most $M^2$ records and ignoring fixed padding factors, successive sizes are $M^{3/2}$, $M^{9/8}$, and $M^{27/32}$, so only a constant number of levels is needed to fit in memory. The four column-sort rounds of one call multiply the total processed volume by only a constant. Across a constant number of levels, the total volume remains a constant multiple of the input size. Summing these costs thus introduces only a constant factor and preserves the exponent $1-\varepsilon$ from address interchange, turning faster rearrangement into faster sorting.

Finally, extend this local algorithm from sizes at most $M^2$ to arbitrary $N>M^2$. Split each record's $\Theta(\log N)$ bits from high to low into groups of $\log_2(\Theta(M))$ bits, each with at most $\Theta(M)$ values. Radix sorting partitions the current array into at most $\Theta(M)$ buckets according to its highest remaining group. To form these buckets on disk, split the array into subarrays of size $\Theta(M^2)$ and apply the local algorithm to each, sorting only by the current bit group. Equal group values then form a contiguous interval in each subarray. Read and concatenate these intervals in bucket-value order; this takes only a scan's worth of block transfers and produces at most $\Theta(M)$ buckets. Their sizes depend on the input. Sort buckets below a stopping threshold of $\Theta(M^2)$ records directly with the local algorithm, and recursively partition larger buckets by the next group. If all groups have been processed, every record in the bucket is equal and recursion stops. All buckets at one depth contain at most $N$ records in total, and the number of levels is $O(\log N/\log M)$. Multiplying the local sorting cost per level by the depth gives Theorem~\ref{thm:sorting}. Sections~\ref{sec:sorting} and~\ref{sec:global-sorting} provide the complete reductions and complexity proofs.
\FloatBarrier

\section{A Counterexample to the Multiple-Unicast Conjecture}
\label{sec:circuit-framework}
\label{sec:finite-construction}

We construct a finite linear circuit and matrix labels on its nodes to prove Theorem~\ref{thm:network}. The circuit gives a binary linear network code with common rate arbitrarily close to one, while the labels bound the maximum common fractional routing rate on the same network strictly below one. Section~\ref{sec:interchange} then uses these labels to guide external-memory address-field interchange, turning the communication advantage into fewer block transfers.

The circuit consists of registers and XOR gates, with each register storing one bit. Gates execute in a predetermined order. A gate reads the current values of specified registers and updates its target to an XOR of current values, leaving all other registers unchanged. For example, $y\leftarrow y+x$ replaces $y$ by the XOR of the current values of $y$ and $x$, while preserving $x$.

The conversion to an undirected graph $G=(\mathcal V,\mathcal E)$ is as in Section~\ref{sec:overview-circuit}. Each gate becomes a vertex. Each register connects the gates involving it in execution order, forming a horizontal main wire with input vertex $s_i$ and output vertex $t_i$. Every segment between successive gates or endpoints becomes a unit-capacity undirected edge. Edges belonging to different register wires are retained separately, so parallel edges are allowed. If the circuit recovers all input values at the outputs according to a permutation $\sigma$, with input $i$ recovered at $t_{\sigma(i)}$, create a session $s_i\to t_{\sigma(i)}$ and assign mutually independent messages to all sources. In addition to the data and auxiliary wires that persist throughout, the circuit branches a temporary wire from a register node to transmit certain intermediate values. This wire connects the subsequent XOR gates that use the copy and terminates after its last use. Each segment is also a unit-capacity undirected edge and contributes to the total edge length. Temporary wires carry intermediate values and introduce no independent sources or sessions.

The full circuit must exchange the initial values of paired data registers and restore the auxiliary initial values. Exchanging two values takes three successive XOR updates: XOR the first group into the second, XOR the updated second group back into the first, and XOR the updated first group into the second. A local circuit performs just one of these updates. After a forward invocation, the control and auxiliary labels already fill the eight-dimensional space, and the target labels have dimension seven. The next update must invoke the local circuit again, but cannot shrink these labels back to their initial low dimensions. To keep labels growing by inclusion through all three updates, we provide a separate eight-dimensional space for each update. Arrange each data-register collection as a three-dimensional grid and invoke the local circuit along the three coordinate directions in three rounds. Each round places the local label changes in its current eight-dimensional factor while preserving the space gained in earlier rounds. The global label space is the tensor product of the three factors, of dimension $8^3=512$. Finally, use three copies of the data grid with cyclically offset coordinate orders. They take turns using the same auxiliary registers, reducing their number while preserving monotone growth of their labels.

\subsection{A Finite Linear Circuit}
\label{sec:finite-coding}

We begin with the local circuit constructed by eumemic and DreamingOfClouds~\cite{E8Base555,E8Word589} in the eight-dimensional rational vector space $E=\mathbb Q^8$. This circuit builds on Sussman's $E_8$ unit~\cite{SussmanE82026}. Equip $E$ with the standard inner product. For $U\subseteq E$, write $U^\perp$ for its orthogonal complement and $0$ for the zero subspace. Subspace labels are used only for the routing bound and address rearrangement in Section~\ref{sec:interchange}; they are independent of the bits stored in the registers.

\begin{lemma}[Properties of the local circuit]
\label{lem:local-interface}
Let $v=111$, $R=439$, and $E=\mathbb Q^8$ with the standard inner product. There exist one-dimensional rational subspaces $\ell_t$ for $t\in[v]$, a permutation $\pi$ of $[v]$, and a finite binary linear circuit on $v$ pairs of data registers $x=(x_t)_{t\in[v]}$, $y=(y_t)_{t\in[v]}$ and $R$ auxiliary registers $z=(z_r)_{r\in[R]}$, a total of $2v+R=661$ registers, implementing
\[
 (x,y,z)\longmapsto(x,y+x,z).
\]
Along each data or auxiliary register's horizontal main wire, the subspace labels grow by inclusion, $U_0\subseteq U_1\subseteq\cdots$. The input and output labels are
\[
\begin{array}{c|cc}
 \text{register}&\text{input}&\text{output}\\\hline
 x_t&\ell_t&E\\
 y_t&0&\ell_t^\perp\\
 z_r&0&E
\end{array}
\]
At every XOR gate, all participating registers have the same subspace label. In addition to the main wires, the circuit has nine temporary wires. Each branches from a register node labeled by a seven-dimensional subspace, changes the copy's label to zero, and then connects the zero-labeled XOR gates that use the copy. These gates add the intermediate value to their respective target registers, and the wire terminates after its last use. Throughout this process, the original register retains both its value and its seven-dimensional label. Only the first segment, from dimension seven to zero, contributes a nonzero dimension difference; all remaining segments have zero labels at both endpoints. The total dimension difference of the nine wires is therefore $c=9\times7=63$. Moreover, $\ell_t\perp\ell_{\pi(t)}$ for every $t\in[v]$.
\end{lemma}
\begin{proof}
Since the local circuit has fixed size, its properties can be checked directly by finite linear-algebra computations; these checks have also been machine-verified in the accompanying Lean~4 code. Represent the $661$ initial bits by independent formal basis vectors over $\mathbb F_2$ and apply the XOR gates in order. The final register expressions are exactly $e_{x_t}$, $e_{y_t}+e_{x_t}$, and $e_{z_r}$. By linearity, every initial state is mapped to $(x,y+x,z)$, including arbitrary auxiliary initial values. Row reduction of the rational generating vectors verifies the inclusions along each main wire and equality of all operand labels at each XOR gate. It also verifies that the seven-dimensional output space of $y_t$ is orthogonal to $\ell_t$, hence equals $\ell_t^\perp$. Each temporary wire's nonzero dimension change is from seven to zero, giving total cost $9\times7=63$. Finally, the generating vectors paired by $\pi$ have zero inner product, proving $\ell_t\perp\ell_{\pi(t)}$.
\end{proof}

\paragraph{The reverse invocation.}
The forward circuit of Lemma~\ref{lem:local-interface} maps $(x,y,z)$ to $(x,y+x,z)$. For the middle step of a three-update exchange, we instead need $(x,y,z)\mapsto(x+y,y,z)$, with exactly the same input and output subspace labels.

First regard each read of a temporary copy as a read of the same intermediate value retained in its original register. Let $g_1,g_2,\dots,g_N$ be the resulting sequence of XOR updates to the main registers. Such an update, for example $A\leftarrow A+B$, is its own inverse over $\mathbb F_2$. Executing $g_N,\dots,g_1$ therefore gives the inverse circuit. Since $(u,w,z)\mapsto(u,w+u,z)$ is itself an involution, the reversed sequence still XORs the first input group into the second. Connect $y$ as its first group and $x$ as its second to obtain
\[
 (x,y,z)\longmapsto(x+y,y,z).
\]
For the labels, replace each original node label $U\subseteq E$ by $U^\perp$ in the reversed circuit. Reversal turns $U_0\subseteq U_1\subseteq\cdots\subseteq U_N$ into a decreasing chain; complementation restores the increasing chain $U_N^\perp\subseteq\cdots\subseteq U_1^\perp\subseteq U_0^\perp$. After exchanging input groups and taking complements, the $x_t$ labels grow from $(\ell_t^\perp)^\perp=\ell_t$ to $0^\perp=E$, the $y_t$ labels grow from $E^\perp=0$ to $\ell_t^\perp$, and the $z_r$ labels grow from $E^\perp=0$ to $0^\perp=E$. These are exactly the endpoints in Lemma~\ref{lem:local-interface}.

For the nine temporary wires, the forward circuit had retained an intermediate value in an original register while adding its copy to later targets. When undoing this batch of additions, first make fresh copies from the original registers, which still hold those values. Add the copies to the targets in reverse order, canceling the forward additions. After undoing the batch, discard the copies and continue undoing the earlier operations that generated the retained values. A forward temporary wire from a seven-dimensional $U$ to zero becomes, after complementation, a wire connecting the one-dimensional $U^\perp$ to the eight-dimensional $E$. Its endpoint dimension difference remains $\dim E-\dim U^\perp=7$, so the total additional cost remains $c=63$.

\paragraph{Three tensor rounds.}
Combine the local circuits in the tensor product $\mathcal H=E_1\otimes E_2\otimes E_3$ of three eight-dimensional spaces, where each $E_j$ is a copy of $E=\mathbb Q^8$. The total dimension is $m=8^3=512$. For $a=(a_1,a_2,a_3)\in[v]^3$, define the line
\[
 U_a=\ell_{a_1}\otimes\ell_{a_2}\otimes\ell_{a_3}.
\]
First consider one collection of register pairs $(X_{a_1,a_2,a_3},Y_{a_1,a_2,a_3})$ arranged as a $v\times v\times v$ grid. Since each local invocation acts on $v$ data pairs, invoke the circuit along the three coordinate directions in three rounds:

\textbf{Round 1 ($v^2$ forward invocations):} For each fixed $(a_2,a_3)\in[v]^2$, use the $v$ pairs $x_t=X_{t,a_2,a_3}$, $y_t=Y_{t,a_2,a_3}$, $t\in[v]$, in one forward invocation.

\textbf{Round 2 ($v^2$ reverse invocations):} For each fixed $(a_1,a_3)$, use registers $X_{a_1,t,a_3}$ and $Y_{a_1,t,a_3}$, for $t\in[v]$, in one reverse invocation.

\textbf{Round 3 ($v^2$ forward invocations):} For each fixed $(a_1,a_2)$, use registers $X_{a_1,a_2,t}$ and $Y_{a_1,a_2,t}$, for $t\in[v]$, in one forward invocation.

The forward, reverse, forward sequence gives every data pair the updates
\[
 (X,Y)\longmapsto(X,Y+X)\longmapsto(Y,Y+X)\longmapsto(Y,X),
\]
exchanging its initial values.

For a fixed pair $(X_a,Y_a)$, we want the subspace labels to grow through the following stages:
\begin{center}
\renewcommand{\arraystretch}{1.35}
\begin{tabularx}{\linewidth}{@{}l>{\centering\arraybackslash}X>{\centering\arraybackslash}X@{}}
\toprule
Stage & Label of $X_a$ (dimension) & Label of $Y_a$ (dimension)\\
\midrule
Input & $\ell_{a_1}\otimes\ell_{a_2}\otimes\ell_{a_3}$ ($1$) & $0$ ($0$)\\
After round 1 & $E_1\otimes\ell_{a_2}\otimes\ell_{a_3}$ ($8$) & $\ell_{a_1}^\perp\otimes\ell_{a_2}\otimes\ell_{a_3}$ ($7$)\\
After round 2 & $E_1\otimes E_2\otimes\ell_{a_3}$ ($64$) & $(\ell_{a_1}\otimes\ell_{a_2})^\perp\otimes\ell_{a_3}$ ($63$)\\
After round 3 & $E_1\otimes E_2\otimes E_3=\mathcal H$ ($512$) & $(\ell_{a_1}\otimes\ell_{a_2}\otimes\ell_{a_3})^\perp=U_a^\perp$ ($511$)\\
\bottomrule
\end{tabularx}
\end{center}
How should a local label $V\subseteq E_j$, of dimension between zero and eight, be lifted to a global label $\phi_j(V)\subseteq\mathcal H$ to produce these stages?

In \textbf{round 1}, with $a_2,a_3$ fixed, the dimensions need only grow from $0,1$ to $7,8$. Tensor the local label with the two fixed lines: $\phi_1(V)=V\otimes\ell_{a_2}\otimes\ell_{a_3}$.

In \textbf{round 2}, with $a_1,a_3$ fixed, using only $\ell_{a_1}\otimes V\otimes\ell_{a_3}$ would lose the seven-dimensional space gained along $\ell_{a_1}^\perp$ in round 1, breaking inclusion at the local circuit's zero- and one-dimensional starting labels. It could also reach dimension at most eight, rather than the required $63,64$. Before invoking the round-2 circuit, extend the already gained $\ell_{a_1}^\perp$ across all of $E_2$. This gives a fixed $7\times8\times1=56$ dimensional subspace $\ell_{a_1}^\perp\otimes E_2\otimes\ell_{a_3}$. Take its orthogonal direct sum with the varying subspace $\ell_{a_1}\otimes V\otimes\ell_{a_3}$, of dimension at most eight.

In \textbf{round 3}, with $a_1,a_2$ fixed, similarly extend the previously gained $63$ dimensional space $(\ell_{a_1}\otimes\ell_{a_2})^\perp$ across $E_3$. The resulting fixed subspace has dimension $63\times8=504$. Take its orthogonal direct sum with $\ell_{a_1}\otimes\ell_{a_2}\otimes V$, of dimension at most eight.

Thus the three lifting maps are
\begin{equation}
\begin{aligned}
 \phi_1(V)&=V\otimes\ell_{a_2}\otimes\ell_{a_3}\\
 &\qquad\text{($a_2,a_3$ fixed; $V\subseteq E_1$)},\\
 \phi_2(V)&=\bigl(\ell_{a_1}^\perp\otimes E_2\;\oplus\;\ell_{a_1}\otimes V\bigr)\otimes\ell_{a_3}\\
 &\qquad\text{($a_1,a_3$ fixed; $V\subseteq E_2$)},\\
 \phi_3(V)&=(\ell_{a_1}\otimes\ell_{a_2})^\perp\otimes E_3\;\oplus\;\ell_{a_1}\otimes\ell_{a_2}\otimes V\\
 &\qquad\text{($a_1,a_2$ fixed; $V\subseteq E_3$)}.
\end{aligned}
\label{eq:tensor-lift}
\end{equation}
The fixed coordinates are the same at every gate in one invocation, so~\eqref{eq:tensor-lift} defines labels for all its gates and temporary wires. The fixed subspace remains unchanged during the invocation and is orthogonal to the varying subspace. Hence, as $V$ grows or changes along a temporary wire, the dimension difference between the lifted labels is exactly the local dimension difference. Before each round, increase the previous endpoint labels of $X_a,Y_a$ to the local starting labels $\phi_j(\ell_{a_j}),\phi_j(0)$, then execute the circuit to reach $\phi_j(E_j),\phi_j(\ell_{a_j}^\perp)$. These endpoints match each row of the table. After three rounds, $X_a$ has grown from $U_a$ to $\mathcal H$ and $Y_a$ from zero to $U_a^\perp$, each with total dimension increase $m-1=511$.

\paragraph{Sharing auxiliary registers among three copies.}
Using fresh auxiliary registers for every local invocation in all three rounds would require many registers. We instead want each auxiliary register to be used three times, with its label growing by inclusion from zero to the full $512$ dimensional space $\mathcal H$.

Within a single data grid, however, the auxiliary registers used in round 1 cannot directly serve round 2. Every local auxiliary register starts at $V=0$, unlike a data register $x_t$, which starts at $V=\ell_{a_j}$. Substituting into~\eqref{eq:tensor-lift}, the round-2 and round-3 auxiliary starting labels are
\[
\begin{aligned}
 \phi_2(0)&=\ell_{a_1}^\perp\otimes E_2\otimes\ell_{a_3}
       &&(56\text{ dimensions}),\\
 \phi_3(0)&=(\ell_{a_1}\otimes\ell_{a_2})^\perp\otimes E_3
       &&(504\text{ dimensions}).
\end{aligned}
\]
The zero local starting space removes the varying summand. The auxiliary starting label uses the full eight-dimensional space in the current coordinate but the orthogonal complement of the tensor product of the one-dimensional subspaces for the already processed coordinates: $\ell_{a_1}^\perp$ in round 2 and $(\ell_{a_1}\otimes\ell_{a_2})^\perp$ in round 3. For example, a pool used along coordinate 1 in round 1, with the other coordinates fixed to $b,c$, ends at $E_1\otimes\ell_b\otimes\ell_c$. Its full $E_1$ factor cannot fit into the $\ell_{a_1}^\perp$ factor at the start of round 2 in that same grid. To retain inclusion, let this pool continue to serve coordinate 1 in the later rounds, using $\pi$ to choose the other fixed indices. In round 2 it serves a grid whose second step is along coordinate 1; in round 3 it serves a grid whose third step is along coordinate 1.

Use three independent data-grid copies $C_1,C_2,C_3$, with cyclically shifted coordinate orders:
$C_1$ uses directions $(1,2,3)$ in rounds $1,2,3$; $C_2$ uses directions $(2,3,1)$; $C_3$ uses directions $(3,1,2)$.
Execute all copies' first rounds, then all second rounds, then all third rounds. In each round, exactly one copy invokes the local circuit along each coordinate direction.

For each $i\in\{1,2,3\}$ and $b,c\in[v]$, allocate a shared pool $P_i(b,c)$ of $R=439$ auxiliary registers, for $3v^2$ pools in total. Pool $P_i(b,c)$ always serves invocations along coordinate $i$ and serves a different data copy in each round. Interpret coordinate and copy indices cyclically modulo three, and display tensor factors in the order $(i,i+1,i+2)$.

\textbf{Round 1, serving the first step of $C_i$:} Vary coordinate $i$, fixing coordinates $i+1,i+2$ to $b,c$. The auxiliary values are restored, and their labels grow from zero to $E_i\otimes\ell_b\otimes\ell_c$, of dimension eight.

\textbf{Round 2, serving the second step of $C_{i+2}$:} This copy has already processed coordinate $i+2$. Vary coordinate $i$ and fix coordinates $i+1,i+2$ to $b,\pi(c)$. The auxiliary starting label is $E_i\otimes\ell_b\otimes\ell_{\pi(c)}^\perp$, of dimension $56$: it is full in coordinate $i$ and uses the orthogonal complement in the processed coordinate $i+2$. Since $\ell_c\perp\ell_{\pi(c)}$, the round-1 endpoint is contained in this starting label. The invocation ends at $E_i\otimes\ell_b\otimes E_{i+2}$, of dimension $64$.

\textbf{Round 3, serving the third step of $C_{i+1}$:} This copy has processed coordinates $i+1,i+2$. Vary coordinate $i$ and fix the other two coordinates to $\pi(b),\pi(c)$. The auxiliary starting label is $E_i\otimes(\ell_{\pi(b)}\otimes\ell_{\pi(c)})^\perp$, of dimension $504$. Since $\ell_b\perp\ell_{\pi(b)}$, it contains the round-2 endpoint $E_i\otimes\ell_b\otimes E_{i+2}$. The invocation ends at $\mathcal H$, of dimension $512$.

The boundary labels of each register in $P_i(b,c)$ therefore form the inclusion chain
\begin{equation}
\begin{aligned}
0&\subseteq E_i\otimes\ell_b\otimes\ell_c
   &&(8,\text{ end of round }1)\\
 &\subseteq E_i\otimes\ell_b\otimes\ell_{\pi(c)}^\perp
   &&(56,\text{ start of round }2)\\
 &\subseteq E_i\otimes\ell_b\otimes E_{i+2}
   &&(64,\text{ end of round }2)\\
 &\subseteq E_i\otimes(\ell_{\pi(b)}\otimes\ell_{\pi(c)})^\perp
   &&(504,\text{ start of round }3)\\
 &\subseteq\mathcal H
   &&(512,\text{ end of round }3).
\end{aligned}
\label{eq:helper-chain}
\end{equation}
Because $\pi$ is a bijection, every pool is used exactly once per round, covering all local invocations. Each use restores its auxiliary initial values. Over the three uses, every auxiliary label grows from zero to $\mathcal H$, with total increase still only $m=512$.

\begin{lemma}[Permutation of the input bits]
\label{lem:finite-coding}
The operation sequence exchanges the initial values of $X,Y$ within each of the three data copies and restores every auxiliary initial value. The full circuit has
\[
 W=6v^3+3v^2R=24{,}432{,}543
\]
independent inputs. Temporary wires introduce no independent inputs or additional sessions.
\end{lemma}
\begin{proof}
The three-round transformation $(X,Y)\mapsto(Y,X)$ holds for every data pair in each copy, and every auxiliary pool is restored before its next use. The three data copies contribute $6v^3$ inputs and the shared pools contribute $3v^2R$ inputs, giving $W$.
\end{proof}

\FloatBarrier
\subsection{Matrix Labels and Total Edge Length}
\label{sec:finite-labels}

Section~\ref{sec:finite-coding} assigns every gate and input or output endpoint a subspace of the $m=512$ dimensional rational space $\mathcal H$. To apply the routing framework of Section~\ref{sec:overview-circuit}, convert each node's subspace $U_u$ into the orthogonal projection $M_u=P_{U_u}$. This is an $m\times m$ rational matrix with row space $U_u$. The routing bound uses ranks of matrix differences. For $U\subseteq V$,
\begin{equation}
 P_V-P_U=P_{V\cap U^\perp},\qquad
 \rank(P_V-P_U)=\dim V-\dim U.
 \label{eq:nested-projections}
\end{equation}
Thus an edge with nested endpoint labels has projection-difference rank equal to their dimension difference.

We now sum these differences over all edges of $G=(\mathcal V,\mathcal E)$:
\[
 S=\sum_{(u,w)\in\mathcal E}\rank(M_w-M_u)
  =\sum_{(u,w)\in\mathcal E}\lvert\dim U_u-\dim U_w\rvert.
\]
Every edge has nested endpoint subspaces in one order or the other. Along a data or auxiliary main wire, labels grow as $U_0\subseteq U_1\subseteq\cdots\subseteq U_K$. Summing the successive differences telescopes to $\dim U_{\mathrm{out}}-\dim U_{\mathrm{in}}$. At an XOR gate, all participating registers have the same label, so no additional length arises there. Each local invocation also has nine temporary wires. In a forward invocation, each wire's nonzero cost is the local dimension change from seven to zero; in a reverse invocation, it is from one to eight. The lifts $\phi_j$ preserve inclusion and these differences, each equal to seven, while all other segments retain the same label. Hence every invocation contributes temporary-wire cost $c=9\times7=63$.

The contributions are summarized below.
\begin{center}
\renewcommand{\arraystretch}{1.3}
\begin{tabularx}{\linewidth}{@{}>{\hsize=1.0\hsize}Y>{\hsize=0.7\hsize}Y>{\hsize=0.9\hsize}Y>{\hsize=1.4\hsize}Y@{}}
\toprule
Component & Number & Cost per item & Saving relative to $Wm$\\
\midrule
Data main wires & $6v^3$ & $m-1=511$ & Save $1$ each; total saving $6v^3=8{,}205{,}786$\\
Auxiliary main wires & $3v^2R$ & $m=512$ & No difference\\
Temporary wires, grouped by invocation & $9v^2$ invocations & $c=63$ per invocation & Extra cost $9v^2c=6{,}986{,}007$\\
\bottomrule
\end{tabularx}
\end{center}

\begin{lemma}[Total edge length and ranks of source--sink matrix differences]
\label{lem:finite-labels}
The ranks of the source--sink matrix differences over all $W=6v^3+3v^2R$ independent sessions sum to
\[
 Wm=3v^2(2v+R)m=12{,}509{,}462{,}016.
\]
The sum of endpoint dimension differences over all network edges is
\begin{equation}
\begin{aligned}
 S&=6v^3(m-1)+3v^2Rm+9v^2c\\
  &=Wm-3v^2(2v-3c)=12{,}508{,}242{,}237<Wm.
\end{aligned}
\label{eq:network-rank-budget}
\end{equation}
For every session $s_i\to t_{\sigma(i)}$, the endpoint spaces satisfy $\dim(U_{s_i}+U_{t_{\sigma(i)}})=m=512$, and their matrix difference has full rank:
\begin{equation}
 \rank_{\mathbb Q}(M_{t_{\sigma(i)}}-M_{s_i})=m.
 \label{eq:network-endpoints}
\end{equation}
\end{lemma}
\begin{proof}
The three copies have $3v^3$ data pairs, hence $6v^3$ main wires. Each $X_a$ label grows from the line $U_a$ to $\mathcal H$, a dimension increase of $m-1$. Each $Y_a$ label grows from zero to $U_a^\perp$, also an increase of $m-1$. Their total contribution is $6v^3(m-1)$. Since the circuit exchanges the initial values of $X_a,Y_a$, the source label $U_a$ of input $X_a$ is paired with the output label $U_a^\perp$ of $Y_a$. These spaces sum to $\mathcal H$. Their projection difference is $I_m-2P_{U_a}$, where $I_m$ is the identity on $\mathcal H$; its square is $I_m$, so it is invertible and has rank $m$. Conversely, input $Y_a$ has label zero and exits at $X_a$ with label $\mathcal H$, giving full-dimensional sum and matrix difference $I_m$.

The $3v^2$ shared auxiliary pools contain $3v^2R$ main wires. By~\eqref{eq:helper-chain}, each grows from zero to $\mathcal H$, contributing $m$, for a total $3v^2Rm$. Each auxiliary initial value is restored on its own wire, so its source--sink spaces are zero and $\mathcal H$, again with full-dimensional sum and matrix difference $I_m$. Finally, each of the three copies executes three rounds of $v^2$ invocations, for $9v^2$ invocations in total. Their temporary-wire dimension differences are preserved by the lifts and contribute $c=63$ per invocation, or $9v^2c$ altogether. Summing gives $S=Wm-3v^2(2v-3c)$. For $v=111$ and $c=63$, we have $2v-3c=33>0$, so the net saving is $33\times3v^2=1{,}219{,}779>0$ and $S<Wm$.
\end{proof}

The fractional reduction of total edge length $S$ relative to $Wm$ is therefore
\begin{equation}
 \eta:=1-\frac{S}{Wm}=\frac{33}{338{,}432}>0,
 \label{eq:finite-saving}
\end{equation}
which is the gap in Theorem~\ref{thm:network}.

\subsection{Proof of Theorem~\ref{thm:network}}
\label{sec:network}
\label{sec:finite-saving}

\begin{proof}[Proof of Theorem~\ref{thm:network}]
Consider the undirected graph $G=(\mathcal V,\mathcal E)$ generated by the full circuit and its $W=24{,}432{,}543$ independent sessions. Every main-wire segment and every segment of a temporary wire corresponds to a unit-capacity undirected edge. Transmitting one bit along each edge in circuit order and executing the XOR operations over $\mathbb F_2$ recovers every source bit at its designated sink by Lemma~\ref{lem:finite-coding}.

Pipelining independent executions makes the common coding rate arbitrarily close to one. Let one execution take a fixed $D$ rounds. Give every source $k$ independent bits, and delay execution of batch $j$ by $j$ rounds for $j=0,\dots,k-1$. Each edge transmits only once in each batch, so it carries at most one bit per round in the pipeline, while each batch retains its causal order. All messages arrive within $k+D$ rounds, giving common rate at least $k/(k+D)$. Letting $k\to\infty$ yields $R_{\mathrm{code}}\ge1$. The unit-capacity cut at any source leaf gives $R_{\mathrm{code}}\le1$, so $R_{\mathrm{code}}=1$.

For routing, suppose a feasible fractional scheme has common rate $\lambda$. Each session $s_i\to t_{\sigma(i)}$ has total path flow $\lambda$, and the sum of all sessions' flows through any undirected edge is at most its unit capacity. Set $\Delta(e)=\rank(M_w-M_u)=\lvert\dim U_u-\dim U_w\rvert$ for $e=uw$. By Lemma~\ref{lem:finite-labels}, every session has source--sink matrix difference of rank $m$. Along any such path $p=(u_0,\dots,u_q)$, the successive differences telescope, and rank subadditivity gives
\[
 m=\rank\!\left(\sum_{j=1}^q(M_{u_j}-M_{u_{j-1}})\right)
 \le\sum_{j=1}^q\rank(M_{u_j}-M_{u_{j-1}})
 =\sum_{e\in p}\Delta(e).
\]
Reversing an edge changes only the sign of its matrix difference, not its rank. Weight this inequality by session $i$'s path flow $f_i(p)$, sum over all sessions, exchange the summations, and use the edge capacities:
\[
\begin{aligned}
 Wm\lambda
 &\le\sum_{i\in[W]}\sum_{p\in\mathcal P_i}f_i(p)\sum_{e\in p}\Delta(e)\\
 &=\sum_{e\in\mathcal E}\Delta(e)
    \sum_{\substack{i\in[W],\ p\in\mathcal P_i\\e\in p}}f_i(p)
 \le\sum_{e\in\mathcal E}\Delta(e)=S.
\end{aligned}
\]
Thus every feasible common routing rate satisfies
\begin{equation}
 \lambda\le\rho:=\frac{S}{Wm}
 =1-\frac{33}{338{,}432}<1.
 \label{eq:network-route-bound}
\end{equation}

To remove parallel edges, insert a distinct new vertex $x_e$ into every edge $uw$, with $M_{x_e}=M_u$. Both replacement edges have capacity one, and their lengths sum to the original edge length. The new vertices only forward information, so the pipeline still has a fixed startup delay and the routing bound is unchanged. The resulting graph has no loops or parallel edges and retains the routing bound and $R_{\mathrm{code}}=1$. Continue to write $D$ for the startup delay after subdivision. Choosing
\begin{equation}
 k>\frac{D\rho}{1-\rho}
 \label{eq:network-blocklength}
\end{equation}
gives a finite-round binary linear protocol with common rate $k/(k+D)>\rho$. Hence a coding-achievable rate strictly exceeds the maximum fractional routing rate, and
\[
 R_{\mathrm{route}}\le1-\frac{33}{338{,}432}<1=R_{\mathrm{code}}.
\]
This proves the theorem.
\end{proof}

The next section uses $S<Wm$ for the same circuit to reduce the total data volume processed by recursive address-interchange subproblems. The inequality $S/W<m^{1-\varepsilon}$ below determines the recursion exponent.

\begin{corollary}[Fixed parameters for the recursion]
\label{prop:finite-circuit}
\label{lem:refined-exponent}
Take $m=512$, $W=24{,}432{,}543$, $S=12{,}508{,}242{,}237$, $\varepsilon=1.5\times10^{-5}$, and $\tau=1-\varepsilon$. The circuit permutes the input bits and has the matrix-label properties proved above, and
\begin{equation}
 \eta=\frac{33}{338432},\qquad \frac SW<m^\tau.
 \label{eq:finite-gap}
\end{equation}
\end{corollary}
\begin{proof}
Appendix~\ref{app:parameters} proves $\varepsilon\ln m<\eta$ by rational comparisons. Since $1-x<e^{-x}$ for $x>0$,
$S/W=m(1-\eta)<m\exp(-\varepsilon\ln m)=m^{1-\varepsilon}$.
\end{proof}

\section{Address-Field Interchange}
\label{sec:memory-cutoff-interchange}
\label{sec:interchange}

In the Columnsort framework introduced in Section~\ref{sec:overview-rearrangement}, rearrangement between columns preserves array values and computes each element's new storage position from its current row and column indices. When the array size is a power of two, each position can be represented by a binary address. Together with sequential scans, the fixed row-and-column rearrangements used here reduce to a constant number of \textbf{interchanges of two equal-width fields in a binary address}: two contiguous groups of address bits exchange places, while all other bits remain fixed. This section uses the finite linear circuit constructed in Section~\ref{sec:circuit-framework} to accelerate this basic operation. Section~\ref{sec:sorting} combines it with sequential scans and in-memory sorting to sort at most $M^2$ records; Section~\ref{sec:global-sorting} extends this to arbitrary input sizes.

\subsection{The interchange operation}
\label{sec:interchange-operation}

Consider a bit array stored contiguously on disk. Represent each bit position by five coordinates, ordered from most to least significant. The two coordinates to be interchanged are represented by $u$-bit binary fields $t,d\in[2^u]$. The remaining coordinates fall into three parts in address order: the higher-order coordinate before the two fields is the \textbf{prefix} $p\in[P]$; the coordinate between them is the \textbf{intervening coordinate} $g\in[G]$; and the lower-order coordinate after them is the \textbf{suffix} $z\in[\ell]$. Here $P,G,\ell$ may be arbitrary positive integers, not necessarily powers of two. Thus every bit position is uniquely represented by a tuple
\[
 (p,t,g,d,z)\in[P]\times[2^u]\times[G]\times[2^u]\times[\ell].
\]
The array contains $V=PG\ell2^{2u}$ bits, stored on disk in lexicographic order of $(p,t,g,d,z)$, with the rightmost suffix coordinate $z$ varying fastest.

An \textbf{address-field interchange} moves the bit at $(p,t,g,d,z)$ to $(p,d,g,t,z)$: it interchanges the values of $t$ and $d$ while leaving the prefix $p$, intervening coordinate $g$, and suffix $z$ unchanged. Fixing a prefix $p\in[P]$ selects a contiguous subarray containing all values of $(t,g,d,z)$, called a \textbf{prefix segment}. The number $\ell$ of suffix values need not equal the word width $w$: the interchange may move only some bits within an integer, whereas if $z\in[\ell]$ contains all bits of an input element, leaving $z$ unchanged moves the element as a whole. Figure~\ref{fig:address-exchange} illustrates field width $u=2$: the element at $(p,01_2,g,10_2,z)$ moves to $(p,10_2,g,01_2,z)$, retaining its value $37$. Below, a \textbf{stream} means a sequence of bits stored on disk.

\begin{figure}[!htbp]
\centering
\begin{tikzpicture}[font=\normalsize]
\node[smalltag] at (5.5,.70) {Address};
\node[smalltag] at (12,.70) {Record value};
\node[datacell,minimum width=1.75cm,minimum height=.65cm,fill=white] (a0f0) at (1.8,0.0) {$ p $};
\node[datacell,minimum width=1.75cm,minimum height=.65cm,fill=lightblue] (a0f1) at (3.7,0.0) {$ 01 $};
\node[datacell,minimum width=1.75cm,minimum height=.65cm,fill=white] (a0f2) at (5.6,0.0) {$ g $};
\node[datacell,minimum width=1.75cm,minimum height=.65cm,fill=warmfill] (a0f3) at (7.5,0.0) {$ 10 $};
\node[datacell,minimum width=1.75cm,minimum height=.65cm,fill=white] (a0f4) at (9.4,0.0) {$ z $};
\node[datacell,minimum width=.85cm,minimum height=.65cm,fill=greenfill] (v0) at (12,0.0) {$37$};
\node[datacell,minimum width=1.75cm,minimum height=.65cm,fill=white] (a1f0) at (1.8,-1.85) {$ p $};
\node[datacell,minimum width=1.75cm,minimum height=.65cm,fill=warmfill] (a1f1) at (3.7,-1.85) {$ 10 $};
\node[datacell,minimum width=1.75cm,minimum height=.65cm,fill=white] (a1f2) at (5.6,-1.85) {$ g $};
\node[datacell,minimum width=1.75cm,minimum height=.65cm,fill=lightblue] (a1f3) at (7.5,-1.85) {$ 01 $};
\node[datacell,minimum width=1.75cm,minimum height=.65cm,fill=white] (a1f4) at (9.4,-1.85) {$ z $};
\node[datacell,minimum width=.85cm,minimum height=.65cm,fill=greenfill] (v1) at (12,-1.85) {$37$};
\draw[figarrow,draw=gray!50,densely dashed] (a0f0.south) -- (a1f0.north);
\draw[figarrow,draw=gray!50,densely dashed] (a0f2.south) -- (a1f2.north);
\draw[figarrow,draw=gray!50,densely dashed] (a0f4.south) -- (a1f4.north);
\draw[figarrow] (a0f1.south) -- (a1f3.north);
\draw[figarrow,draw=warmink] (a0f3.south) -- (a1f1.north);
\draw[figarrow,draw=greenink] (v0.south) -- (v1.north);
\end{tikzpicture}
\caption{Interchanging two address fields preserves the record value. The fields are $01$ and $10$; the other coordinates $p,g,z$ remain unchanged.}
\label{fig:address-exchange}
\end{figure}
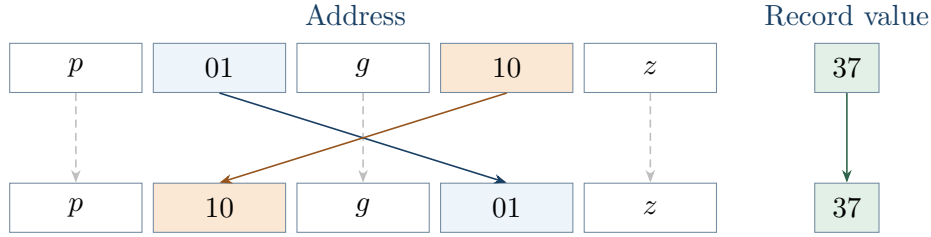

Throughout this section, the array dimensions are bounded by a fixed polynomial in the global input size $N$, so dimension parameters and disk addresses occupy a constant number of $w$-bit machine words.

Write $a=\log(M/B)$ and $\tau=1-\varepsilon$, where $\varepsilon=1.5\times10^{-5}$ is a constant. Here $M/B$ is the number of disk blocks that fit simultaneously in memory, so $a=\log(M/B)$ measures the field width for which the memory buffers can perform an interchange directly.

\begin{lemma}[Block transfers for address-field interchange]
\label{lem:memory-cutoff-interchange}
\label{lem:memory-cutoff-swap}
There is an absolute constant $C_{\rm abs}$ such that, for $M\ge\max\{2B,C_{\rm abs}\}$, an address-field interchange of width $u\ge1$ on an array of $V=PG\ell2^{2u}$ bits takes at most
\begin{equation}
 O\!\left(1+\frac{V}{Bw}
       \left[1+\left(\frac{u}{a}\right)^\tau\right]\right)
 \label{eq:memory-cutoff-interchange}
\end{equation}
block transfers and uses at most $M$ words of memory. The algorithm is deterministic, and its execution steps are determined entirely by the array dimensions, independently of the stored values. Its disk workspace is polynomial in $N$, and all dimension parameters and disk addresses occupy a constant number of machine words of width $w=\Theta(\log N)$.
\end{lemma}

Sections~\ref{sec:coding-recursion} and~\ref{sec:recursion-bound} give the full proof of Lemma~\ref{lem:memory-cutoff-interchange} in stages. We draw on OpenAI's method of recursively implementing address-field interchange with a finite linear circuit on a multitape Turing machine~\cite[Proposition~11 and Lemma~12]{OpenAI2026}. Combining it with our 512-dimensional linear circuit from Section~\ref{sec:circuit-framework} brings this method to the external-memory block-transfer model: Section~\ref{sec:coding-recursion} reduces a long-field interchange to interchanges of shorter fields and uses the circuit to reduce the total data volume processed by these subproblems; Section~\ref{sec:recursion-bound} stops the recursion when the fields are narrow enough for the memory buffers to handle directly and bounds the total number of block transfers.

Once memory reaches the fixed threshold $C_{\rm abs}$, Lemma~\ref{lem:memory-cutoff-interchange} requires only $M\ge2B$, whereas scanning multiple disk streams requires a fixed number of block buffers in memory at once. We therefore first assume $M/B\ge R_0$ in Sections~\ref{sec:coding-recursion}--\ref{sec:recursion-bound}, for a sufficiently large fixed constant $R_0$; the constant $C_{\rm abs}$ accommodates the fixed circuit's control information. Lemma~\ref{lem:io-normalization} in Appendix~\ref{app:block-simulation} shows that reducing the logical memory and block sizes by fixed constant factors extends the result to every $M\ge\max\{2B,C_{\rm abs}\}$ without changing the asymptotic number of block transfers.

\subsection{The Coding Recursion}
\label{sec:coding-recursion}

By Corollary~\ref{prop:finite-circuit}, the finite circuit constructed in Section~\ref{sec:circuit-framework} has fixed parameters $m=512$, $W=24{,}432{,}543$, and $S=12{,}508{,}242{,}237$, satisfying
\begin{equation}
 \lambda:=\frac{S}{W}<m^\tau.
 \label{eq:cutoff-rank-deficit}
\end{equation}
The inequality above lets us reduce the total data volume processed by recursive subproblems. Divide the original array equally by prefix among $W$ streams of $V/W$ bits and feed them into the circuit. Changing address order along an edge of rank $r$ requires only $r$ short-field interchanges, each of width $1/m$ of the original field. Summing over all edges reduces one complete field interchange to $S$ short-field interchanges, lowering the total data volume of the subproblems from $mV$ under naive divide-and-conquer to $(S/W)V=\lambda V<m^\tau V$.

Array addresses must be integers, whereas the projection matrices in Section~\ref{sec:circuit-framework} have rational entries $a/b$. During the recursion, we therefore represent the fields in a fixed odd prime radix $q$, so that each short field of $f$ base-$q$ digits corresponds to an element of the residue ring $\mathbb Z/q^f\mathbb Z$. The circuit and the matrix factorizations below involve only finitely many rational numbers. We can therefore choose $q$ in advance so that it divides none of their denominators and none of the numerators of the nonzero diagonal entries of the invertible triangular factors. Interpreting each rational number $a/b$ as the modular-inverse product $ab^{-1}\bmod q^f$ then preserves all matrix identities and invertibility over $\mathbb Z/q^f\mathbb Z$ for every $f\ge1$. Arithmetic modulo $q^f$ applies only to address coordinates: the data stored on disk remain binary bits, and the circuit gates still perform bitwise XOR.

\begin{lemma}[One level of the coding recursion]
\label{lem:coded-recursion}
Suppose the two fields $H,D$ to be interchanged each contain $e=mf$ base-$q$ digits, the number $P$ of prefix segments is divisible by $W$, and the array contains $V$ bits. A complete address-field interchange of width $e$ reduces to $S$ address-field interchanges of width $f$, each processing $V/W$ bits. Apart from these recursive calls, the current level requires only $O(1+V/(Bw))$ additional block transfers and a fixed number of block buffers.
\end{lemma}
\begin{proof}
Divide each long field of width $e=mf$ into $m$ short fields of width $f$, so that $H$ and $D$ are $m$-dimensional column vectors over the ring $\mathbb Z/q^f\mathbb Z$. Throughout the following operations, the intervening coordinate $g$ and the outer coordinates $p,z$ remain unchanged.

\paragraph{Partial swaps.}
For the orthogonal projection matrix $P$ onto a rational subspace, define the following linear permutation $D_P$ of the address vector $(H,D)$, called the partial swap associated with $P$:
\begin{equation}
 D_P\begin{pmatrix}H\\D\end{pmatrix}
 =\begin{pmatrix}I-P&P\\P&I-P\end{pmatrix}\begin{pmatrix}H\\D\end{pmatrix}
 =\begin{pmatrix}(I-P)H+PD\\PH+(I-P)D\end{pmatrix}.
 \label{eq:partial-swap}
\end{equation}
This transformation exchanges only the components of $H$ and $D$ projected onto the image of $P$, leaving their components in its orthogonal complement unchanged. Since an orthogonal projection satisfies $P^2=P$, direct matrix multiplication gives $D_P^2=I$. In particular, $D_0=I$ leaves the address unchanged, whereas $D_I(H,D)=(D,H)$ completely interchanges $H$ and $D$. If two subspaces satisfy $U\subseteq V$ and have orthogonal projections $P_U$ and $P_V$, respectively, then $P_UP_V=P_VP_U=P_U$, which gives
\[
 D_{P_V}D_{P_U}=D_{P_V-P_U}.
\]
Since every partial swap is its own inverse, the reverse transition $D_{P_U}D_{P_V}$ also equals $D_{P_V-P_U}$. Thus, when the subspaces at the endpoints of a circuit edge are nested, changing a stream's disk storage order from one endpoint's projection label to the other's requires just one partial swap defined by their projection difference $P_V-P_U$. The rank $\rank(P_V-P_U)=\dim V-\dim U$ is exactly the subspace dimension difference counted in the sum $S$ in Section~\ref{sec:circuit-framework}.

\paragraph{A rank-$r$ projection requires $r$ short-field interchanges.}
Let $P$ be the projection difference on an edge, with a rational image subspace of dimension $r=\rank P$. Applying Gram--Schmidt to a rational basis of this subspace, but omitting the final normalization by vector norms, yields pairwise orthogonal nonzero rational basis vectors $u_1,\dots,u_r\in\mathbb Q^m$. For each basis vector $u_k$, the matrix of the orthogonal projection onto the line it spans is
\[
 P_k=\frac{u_ku_k^{\mathsf T}}{u_k^{\mathsf T}u_k}.
\]
The denominator $u_k^{\mathsf T}u_k$ is a sum of squares of rational numbers and requires no square root, so each $P_k$ is a rational orthogonal projection of rank one. The pairwise orthogonality of $u_1,\dots,u_r$ gives $u_k^{\mathsf T}u_l=0$ when $k\ne l$, and hence $P_kP_l=0$. Moreover, the projection onto the entire $r$-dimensional subspace is the sum of the projections onto these $r$ orthogonal directions: $P=P_1+P_2+\dots+P_r$. Substituting $P_kP_l=0$ into~\eqref{eq:partial-swap} and multiplying directly shows that composing the partial swaps defined by these rank-one projections exchanges the components in all of these orthogonal directions:
\[
 D_P=D_{P_1}D_{P_2}\cdots D_{P_r}.
\]
Thus, when $\rank P=r$, $D_P$ decomposes into $r$ partial swaps, each exchanging the components in one orthogonal direction. For a nonzero rational vector $u\in\mathbb Q^m$, let $P=uu^{\mathsf T}/(u^{\mathsf T}u)$ be the projection onto the one-dimensional subspace it spans. We now show how one short-field interchange and a constant number of sequential scans implement the corresponding partial swap $D_P$.

For this purpose, set $\xi=(u,-u)^{\mathsf T}\in\mathbb Q^{2m}$ and write $\nu=\xi^{\mathsf T}\xi=2u^{\mathsf T}u$, a nonzero rational number. Substituting $P=\frac{uu^{\mathsf T}}{u^{\mathsf T}u}=\frac{2}{\nu}uu^{\mathsf T}$ into the block matrix in~\eqref{eq:partial-swap} rewrites the $2m\times2m$ matrix $D_P$ as the reflection in $\mathbb Q^{2m}$ with normal vector $\xi$:
\[
 D_P=\begin{pmatrix}
 I-\frac{2}{\nu}uu^{\mathsf T}&\frac{2}{\nu}uu^{\mathsf T}\\
 \frac{2}{\nu}uu^{\mathsf T}&I-\frac{2}{\nu}uu^{\mathsf T}
 \end{pmatrix}
 =I_{2m}-\frac{2}{\nu}\xi\xi^{\mathsf T}.
\]
For any vector $v\in\mathbb Q^{2m}$ orthogonal to $\xi$, we have $\xi^{\mathsf T}v=0$ and therefore $D_Pv=v$, so the transformation fixes $v$. For $\xi$ itself, the identity $\xi^{\mathsf T}\xi=\nu$ gives $D_P\xi=-\xi$.

Using this property, we can construct an invertible lower-triangular coordinate transformation $L$ that turns the reflection $D_P$ into a simple exchange of two coordinate positions. Since $u\ne0$, the vector $\xi=(u,-u)^{\mathsf T}$ has a nonzero entry among its first $m$ coordinates and among its last $m$ coordinates. Let $i$ and $j$ be the indices of its first and last nonzero entries, respectively. Thus $i$ lies in the first field $H$, and $j$ lies in the second field $D$. Set $\alpha=-\nu/(2\xi_j)$, and let $e_k$ denote the $k$th standard basis vector of $\mathbb Q^{2m}$. Define the columns of the $2m\times2m$ matrix $L$ by
\[
 L_i=\xi+\alpha e_j,\qquad L_j=\alpha e_j,\qquad
 L_k=e_k-\frac{\xi_k}{\xi_j}e_j\quad(i<k<j),
\]
and set $L_k=e_k$ for the remaining indices $k<i$ or $k>j$. Since $\xi$ vanishes before index $i$ and after index $j$, each column $L_k$ has zero entries above row $k$, making $L$ lower triangular. Its diagonal entries at rows $i$ and $j$ are the nonzero rationals $\xi_i$ and $\alpha$, respectively, and all other diagonal entries are $1$, so $L$ is invertible.

We now check the effect of left multiplication by $D_P$ on each column of $L$. For column $L_j=\alpha e_j$, the equality $\xi^{\mathsf T}(\alpha e_j)=\alpha\xi_j=-\nu/2$ gives
\[
 D_PL_j=\alpha e_j-\frac{2}{\nu}\xi\left(-\frac{\nu}{2}\right)
 =\alpha e_j+\xi=L_i.
\]
The identity $D_P^2=I$ then immediately gives $D_PL_i=L_j$. Thus, $D_P$ exchanges columns $i$ and $j$ of $L$. For every other column $L_k$, direct evaluation of the inner product gives $\xi^{\mathsf T}L_k=0$, and hence $D_PL_k=L_k$. Therefore, if $T_{ij}$ denotes the permutation matrix that directly exchanges coordinates $i$ and $j$, then
\begin{equation}
 D_PL=LT_{ij},\qquad D_P=LT_{ij}L^{-1}.
 \label{eq:partial-swap-triangular}
\end{equation}
Equation~\eqref{eq:partial-swap-triangular} shows that applying this partial swap $D_P$ to a data stream requires first applying the lower-triangular transformation $L^{-1}$ to its address coordinates, then invoking one short-field interchange $T_{ij}$ on fields $i$ and $j$, and finally applying $L$. In the middle step $T_{ij}$, the short fields between indices $i$ and $j$, together with the original intervening coordinate $g$, form the new intervening coordinate. The short fields before $i$ join the prefix, and those after $j$ join the suffix. This is therefore exactly an address-field interchange of width $f$.

The matrices $L$ and $L^{-1}$ at the ends of~\eqref{eq:partial-swap-triangular} are invertible lower-triangular matrices with fixed dimension $2m$ and fixed entries. Our preselected odd prime $q$ divides none of their denominators or the numerators of the diagonal entries $\xi_i,\alpha$, so both matrices are well defined and invertible modulo $q^f$ for every width $f\ge1$. Since the matrix is lower triangular, compute the new coordinates from last to first. When a new coordinate is computed, all earlier coordinates it uses still have their original values. Each step can be implemented by the rearrangement routines of OpenAI~\cite[Lemma~4]{OpenAI2026}; the intervening coordinate $g$ and the suffix remain unchanged. Since the dimension is fixed, these updates require no recursive calls and take $O(V/W)$ basic tape operations in total, with a fixed number of block buffers. Consequently, the partial swap $D_P=LT_{ij}L^{-1}$ costs one short-field interchange $T_{ij}$ of width $f$, plus a constant number of sequential scans; these scans take $O(V/W)$ basic tape operations in total.

\paragraph{Distributing the streams and executing the circuit.}
Fixing a prefix $p\in[P]$ identifies a contiguous prefix segment on disk containing all values of $(H,g,D,z)$. By the lemma's hypothesis, the total number $P$ of prefix segments is a multiple of $W$. A sequential scan can therefore distribute these $P$ segments evenly among $W$ streams in cyclic order modulo $W$, placing segment $Wj+i$ at position $j$ in input stream $F_i$. Each of the circuit's $W$ inputs thus receives a stream containing $P/W$ prefix segments and $V/W$ bits, including the inputs assigned to auxiliary registers.

To implement the circuit's XOR gates by sequential disk scans, let $G$ denote the bit sequence computed by the circuit at a node. If the node has subspace label $U$, store $G$ in the order specified by the partial swap $D_{P_U}$, so the sequence on disk is $D_{P_U}G$. Maintain this storage order as the circuit gates are executed in their prescribed order. Write $M_{s_i}$ for the projection matrix labeling source $s_i$ of input $i$, and define the stream supplied to circuit input $i$ by $G_i=D_{M_{s_i}}F_i$. Since $D_{M_{s_i}}^2=I$, applying $D_{M_{s_i}}$ to both sides gives $F_i=D_{M_{s_i}}G_i$. Thus, without any preprocessing, the original input stream $F_i$ already has the storage order prescribed by source label $M_{s_i}$ for circuit input stream $G_i$.

When a stream travels along a register's main wire or a temporary wire from a node labeled $U$ to a node labeled $V$, the circuit guarantees that $U$ and $V$ are nested. Applying the partial swap $D_{\lvert P_V-P_U\rvert}=D_{P_V}D_{P_U}$ associated with their projection difference changes the stream's disk storage order from $D_{P_U}$ to $D_{P_V}$. Here $\lvert P_V-P_U\rvert$ denotes the orthogonal projection onto the larger subspace minus the orthogonal projection onto the smaller subspace. At each XOR gate, all participating streams have the same node label $U$, so the disk streams have undergone the same address permutation $D_{P_U}$. Because a common address permutation commutes with bitwise XOR, one sequential scan of these disk streams, XORing bits at corresponding positions, performs the bitwise XOR required by the circuit while preserving storage order $D_{P_U}$. A stream on a temporary wire is likewise copied at its source node, converted to the target gate's storage order, used in the XOR operation, and then discarded. The number of short-field interchanges required for this storage-order conversion is already included in $S$.

By Lemma~\ref{lem:finite-coding}, when execution reaches the outputs, the circuit's linear XOR operations recover each input stream $G_i$ unchanged at its corresponding sink $t_{\sigma(i)}$. Since the stream there is stored in the order prescribed by the sink's projection label $M_{t_{\sigma(i)}}$, the actual output stream on disk is
\[
 D_{M_{t_{\sigma(i)}}}G_i
 =D_{M_{t_{\sigma(i)}}}D_{M_{s_i}}F_i.
\]
By the source and sink projection labels defined in Section~\ref{sec:finite-labels}, a register $X_a$ in the first group of each data-register pair has source label $M_{s_i}=P_{U_a}$ and corresponding sink label $M_{t_{\sigma(i)}}=I-P_{U_a}$. Their partial swaps compose to $D_{I-P_{U_a}}D_{P_{U_a}}=D_I$. For a register $Y_a$ in the second group, as well as for every auxiliary register, the source and corresponding sink labels are $0$ and $I$, respectively, and again their composition is $D_ID_0=D_I$. Thus the actual output at each of the $W$ sinks is exactly $D_IF_i$, completing the interchange of the long fields $H$ and $D$ within each input stream $F_i$. One final sequential scan returns the output at sink $t_{\sigma(i)}$ to the prefix-segment positions originally assigned to stream $i$. This undoes the exchange of register-stream indices induced by $\sigma$ and completes the address-field interchange on the entire array.

\paragraph{Recursive-call count and I/O cost.}
The subspaces at the two endpoints of every wire segment and temporary wire in the circuit are nested, and the ranks of their projection differences sum to $S$. For a projection difference of rank $r$, the corresponding partial swap decomposes into $r$ partial swaps along individual orthogonal directions. Each calls one short-field interchange of width $f$. The entire circuit therefore makes $S$ child interchange calls, each of width $f$ on $V/W$ bits. Moreover, if a child call interchanges the $i$th and $j$th short fields, the $i-1$ short fields preceding the $i$th field, each with $q^f$ possible values, join the original prefix. The number of prefix segments in that child is thus $(P/W)(q^f)^{i-1}$, a positive integer multiple of $P/W$.

Apart from these $S$ recursive child calls, all operations at the current level---the initial distribution by prefix, stream copying and bitwise XOR at circuit nodes, the lower-triangular coordinate transformations $L$ and $L^{-1}$, and the final assembly of the output streams into the original array---consist of a constant number of sequential reads and writes over the streams. Because the circuit's size and structure are fixed, these reads and writes process $O(V)$ bits in total and require only $O(1+V/(Bw))$ block transfers in external memory; see Appendix~\ref{app:tape-access}. The algorithm executes the $S$ child calls one at a time, writing the suspended parent's state and buffers back to disk. Only the currently executing child uses memory, so recursion control and the current level's stream operations occupy only a fixed number of memory blocks; see Appendix~\ref{app:recursive-suspension}. Summing the costs of all $S$ child calls gives the following recurrence for the current level's total number of block transfers $Q(e,V)$:
\[
 Q(e,V)\le S\,Q(e/m,V/W)+O\!\left(1+\frac{V}{Bw}\right).\qedhere
\]
\end{proof}

\subsection{The Recursion Bound}
\label{sec:recursion-bound}

Each level of the coding recursion reduces the width of the fields to be interchanged by a factor of $m$. Once the field width $e$ is sufficiently small, we can stop the recursion and complete the interchange directly using memory buffers and sequential scans. Specifically, two fields $(H,D)$ of $e$ base-$q$ digits have $q^{2e}$ possible value pairs. To distribute array entries into disk streams according to their $(H,D)$ values in one sequential scan, memory must simultaneously hold one block of $B$ words for each of these $q^{2e}$ streams. Since memory holds $M/B$ blocks in total, we choose the stopping threshold
\[
 t=\left\lfloor\frac18\log_q(M/B)\right\rfloor.
\]
The assumption $M/B\ge R_0$ from Section~\ref{sec:interchange-operation}, for a sufficiently large fixed constant $R_0$, ensures that $t\ge1$. Whenever $e\le t$, the required number of memory blocks satisfies $q^{2e}\le(M/B)^{1/4}$, which is much smaller than the available $M/B$ blocks.

Because each recursive level divides the field width by $m$, we first analyze widths of the form $e=t_0m^k$, where the field width at the bottom of the recursion is
\[
 t_0=m^{\lfloor\log_m t\rfloor}.
\]
Starting from width $e=t_0m^k$, exactly $k$ recursive levels reduce the field width to $t_0\le t$. Since $t_0\le t<mt_0$ and $m,q$ are fixed constants, this bottom-level width satisfies $t_0=\Theta\bigl(\log(M/B)\bigr)=\Theta(a)$.

\begin{lemma}[Interchange for field width $t_0m^k$]
\label{lem:interchange-scales}
If the fields to be interchanged have width $e=t_0m^k$ and the number $P$ of prefix segments is divisible by $W^k$, then the number of block transfers needed for the address-field interchange is at most
\begin{equation}
 Q(e,V)=O\!\left(1+\frac{V}{Bw}\left(\frac{e}{t_0}\right)^\tau\right).
 \label{eq:cutoff-power}
\end{equation}
In particular, when $e\le t$, no recursion is needed, and the interchange takes only $O(1+V/(Bw))$ block transfers.
\end{lemma}
\begin{proof}
As in Section~\ref{sec:interchange-operation}, we assume that the number $M/B$ of blocks that fit in memory exceeds a sufficiently large fixed constant.

\paragraph{Case 1: no recursion ($e\le t$).}
If the array's total bit volume is at most half the memory capacity, that is, $V\le Mw/2$, read the entire array into memory, rearrange it, and write it back to disk. This takes only $O(1+V/(Bw))$ block transfers. If $V>Mw/2$, prepare $q^{2e}$ disk streams, one for each possible value pair $(H,D)$, and keep one block in memory for each stream; here $q^{2e}\le(M/B)^{1/4}$. First scan the input array sequentially and distribute its entries among these streams according to their $(H,D)$ values. Then read from the $q^{2e}$ streams in the target order $(p,D,g,H,z)$ to assemble the output array. A stream writes a block to disk only after accumulating $B$ words. Thus, in addition to the $O(V/(Bw))$ full-block reads and writes charged to the total data volume, each stream incurs at most $O(1)$ extra block transfers for the partial block at the end of its reads and writes. The total extra cost over the $q^{2e}$ streams is at most $O(q^{2e})\le O\bigl((M/B)^{1/4}\bigr)$. Since $V>Mw/2$ and $M/B$ is sufficiently large, we have $(M/B)^{1/4}<M/(2B)<V/(Bw)$, so this extra cost is absorbed by $O(V/(Bw))$. Hence, when $e\le t$, the interchange takes only $O(1+V/(Bw))$ block transfers.

\paragraph{Case 2: multiple recursive levels ($e=t_0m^k$, $k\ge1$).}
By Lemma~\ref{lem:coded-recursion}, each recursive level divides the field width by $m$ and makes the prefix count a positive integer multiple of the parent's prefix count divided by $W$. Since the initial prefix count $P$ is divisible by $W^k$, after $k$ levels the field width reaches $t_0\le t$, reducing the problem to Case~1. During the recursion, if a subproblem's volume satisfies $V'\le Mw/2$, read it into memory and complete its interchange directly, without expanding it further. Every subproblem that is expanded therefore has volume $V'>Mw/2$. Because $M/B$ is sufficiently large, exceeding the fixed constant $2W$, we have $V'>Mw/2\ge WBw$. Each child consequently has volume $V'/W$ of at least one full disk block of $Bw$ bits, so the $O(1+V'/(WBw))$ block-transfer overhead per child in the recurrence of Lemma~\ref{lem:coded-recursion} simplifies to $O(V'/(WBw))$.

We now count the total block transfers over all $k$ recursive levels. Each level replaces one problem by $S$ children, each with $1/W$ of its data volume. Thus level $j$, for $0\le j\le k$, has at most $S^j$ subproblems, each of volume $V/W^j$, and their total data volume is at most
\[
 S^j\cdot\frac{V}{W^j}=\left(\frac{S}{W}\right)^jV=\lambda^jV.
\]
At level $k$, where the field width is $t_0\le t$, Case~1 likewise bounds each leaf's block transfers by a constant times its own data volume measured in blocks. Summing the block transfers from level $0$ through level $k$, and using $\lambda=S/W>1$, gives a geometric sum dominated by its final term:
\[
 Q(e,V)\le\sum_{j=0}^k O\!\left(\lambda^j\frac{V}{Bw}\right)+O(1)
 =O\!\left(1+\lambda^k\frac{V}{Bw}\right).
\]
By~\eqref{eq:cutoff-rank-deficit}, $\lambda<m^\tau$, so $\lambda^k<(m^k)^\tau=(e/t_0)^\tau$. Substitution proves~\eqref{eq:cutoff-power}.
\end{proof}

\begin{proof}[Proof of Lemma~\ref{lem:memory-cutoff-interchange}]
First, Appendix Lemma~\ref{lem:arbitrary-fields} extends Lemma~\ref{lem:interchange-scales} from widths $e=t_0m^k$ and prefix counts divisible by $W^k$ to every positive integer width $e\ge1$ and every prefix count $P\ge1$. The idea is to extract several high-order digits from each field and move them into the prefix, then split the remaining width into subfields of widths $t_0m^j$ and apply Lemma~\ref{lem:interchange-scales} to interchange them in sequence. The widest subfield dominates the total cost of these interchanges; any residual short part of width less than $t_0$ is interchanged by sequential scans. At the same time, padding with at most the original array's volume of blank data makes the prefix count a multiple of the required $W^j$. Thus, interchanging base-$q$ fields of any width $e$ takes at most
\[
 Q(e,V)=O\!\left(1+\frac{V}{Bw}\left[1+\left(\frac{e}{t_0}\right)^\tau\right]\right)
\]
block transfers.

Next, we convert the interchange of binary fields of width $u\ge1$ in Lemma~\ref{lem:memory-cutoff-interchange} into an interchange of base-$q$ fields. Choose the smallest positive integer $e$ such that $q^e\ge2^u$. Since $q$ is fixed, we have $q^e<q\cdot2^u$ and $e=\Theta(u)$. In one sequential scan, extend the ranges of the two fields $t,d\in[2^u]$ in the disk array from $[2^u]$ to $[q^e]$, treating the new coordinate positions, where $t\ge2^u$ or $d\ge2^u$, as blank padding. The expanded volume $V'$ is at most $q^2=O(1)$ times the original volume $V$. Perform the base-$q$ address-field interchange of width $e$ on this expanded array. The operation moves an entry from $(p,t,g,d,z)$ to $(p,d,g,t,z)$, so every genuine array entry whose coordinates initially satisfy $t,d<2^u$ still satisfies $d,t<2^u$ after the interchange and remains within the valid coordinate range of the original array. Finally, a sequential scan removes the padding positions with $t\ge2^u$ or $d\ge2^u$, yielding the target array with its binary fields interchanged.

Substituting $V'=O(V)$, $e=\Theta(u)$, and $t_0=\Theta(a)$ into the bound above gives the block-transfer bound~\eqref{eq:memory-cutoff-interchange} in Lemma~\ref{lem:memory-cutoff-interchange}. Appendix Lemma~\ref{lem:io-normalization} then reduces the logical memory and block sizes by fixed constant factors, extending the memory condition from sufficiently large $M/B$ to every $M\ge\max\{2B,C_{\rm abs}\}$ without changing the asymptotic number of block transfers. Because the circuit structure and transformation matrices are fixed, all read and write positions are uniquely determined by the array dimensions, independently of the stored values. Under the stated dimension bound, all temporary streams, padded arrays, and recursive-call records have polynomial size in $N$, and their dimension parameters and disk addresses occupy a constant number of machine words when $w=\Theta(\log N)$. Together with the buffering arrangements above, this establishes all the claims of Lemma~\ref{lem:memory-cutoff-interchange}.
\end{proof}

\section{Local Sorting}
\label{sec:sorting}

Section~\ref{sec:interchange} gives a faster method for address-field interchange, whose movement rules are determined entirely by the elements' current indices. In sorting, however, the target positions depend on the values of the input records. To turn faster fixed address rearrangements into faster sorting, we need a sorting framework that separates comparisons depending on record values from global data movement following a fixed pattern.

Leighton's Columnsort algorithm~\cite{Leighton1985,ColumnsortNotes} provides this bridge. It views the input array as a matrix with many more rows than columns and sorts it by alternating four rounds of independent column sorting with four rounds of data movement among fixed positions. All communication between columns occurs in the fixed operations at even-numbered steps and is independent of the input values. Our approach is to choose the matrix dimensions so that global data movement in the eight-step Columnsort procedure reduces to a constant number of address-field interchanges from Section~\ref{sec:interchange} and sequential scans. The column sorts at odd-numbered steps are handled recursively by the same algorithm. Each call reduces column size to within a constant factor of the $3/4$ power of the array size. If the initial number of records satisfies $n\le M^2$, a constant number of recursion levels therefore makes every column small enough to sort in memory, giving local sorting with the same exponent $1-\varepsilon$ in the complexity bound.

We call each element to be sorted a record. Each input integer occupies $w=\Theta(\log N)$ bits, where $N$ is the global input size, so adding the two-bit tags used below to mark padding still leaves each record with $O(w)$ bits. A record may also consist of a constant number of integer fields compared lexicographically, with the same $O(w)$-bit size. We retain the notation $a=\log(M/B)$ and $\tau=1-\varepsilon$ from Section~\ref{sec:interchange} and initially assume that both $M$ and $M/B$ exceed sufficiently large fixed constants. Corollary~\ref{cor:parameter-transfer} in the appendix extends the resulting bounds to all $2B\le M$.

\subsection{The Columnsort Reduction}
\label{sec:columnsort-reduction}

To connect the eight-step Columnsort procedure with binary address-field interchange, arrange the records in a matrix with $s$ columns of $r$ records each, stored contiguously by column on disk. We use the following sufficient conditions for Columnsort: $r$ is even, $s$ divides $r$, and $r\ge2s(s+1)$. Take $t\ge3$, set $s=2^t$, and choose the number of rows to be the cube of the number of columns, $r=s^3=2^{3t}$. This choice has two immediate benefits. First, the column length $r$ is exactly the $3/4$ power of the total number of records $rs=s^4$, ensuring rapid reduction of the recursive problem size. Second, since $rs=2^{4t}$, each record's $4t$-bit index splits into four fields of width $t$. In binary, taking the remainder and quotient upon division by $s=2^t$ extracts the lowest field and the highest three fields, respectively. The intercolumn transpose rearrangement therefore becomes a cyclic rotation of these four fields.

\begin{proposition}[The Columnsort reduction]
\label{prop:columnsort-reduction}
Let $s=2^t$ and $r=s^3$, where $t\ge3$. For an array stored on disk as $s$ contiguous columns of $r$ records each, the eight-step Columnsort procedure can be implemented using four rounds of column sorting, six address-field interchanges of width $t$, and a constant number of sequential scans. If all four rounds of column sorting are correct, the entire array is sorted correctly.
\end{proposition}
\begin{proof}
Number one-dimensional positions from zero, so that the record in column $c$ and row $q$ occupies position $j=cr+q$. The algorithm performs the following eight steps. Odd-numbered steps sort each column independently; even-numbered steps perform fixed operations on the global layout.
\begin{enumerate}[label=\arabic*.,itemsep=1pt,topsep=3pt]
\item Sort each of the $s$ columns independently.
\item Apply the transpose rearrangement $\pi$, moving the record at position $j$ to $\pi(j)=(j\bmod s)r+\lfloor j/s\rfloor$. This reads the original matrix by columns and fills the new matrix by rows, distributing the records from each original column evenly among the $s$ new columns.
\item Sort each of the rearranged $s$ columns independently.
\item Apply the inverse permutation $\pi^{-1}$ from step~2, moving the record at position $dr+q$ back to position $sq+d$ and restoring the original column layout.
\item Sort each of the $s$ columns independently again.
\item Prepend $r/2$ minimum padding records $-\infty$, strictly smaller than every genuine record, and append $r/2$ maximum padding records $+\infty$, strictly larger than every genuine record. The array now has $s+1$ columns of $r$ records each, and the original column boundaries move forward by half a column.
\item Sort each of these $s+1$ columns independently.
\item Delete the first and last $r/2$ records, leaving an output array of $rs$ records.
\end{enumerate}

\textbf{Correctness.} When $t\ge3$, we have $s=2^t\ge8$. Thus $r=s^3$ is even, $s$ divides $r$, and $r=s^3\ge8s^2>2s(s+1)$. These parameters satisfy the Columnsort dimension conditions, so correctness of the eight steps follows from Leighton~\cite{Leighton1985}. To represent $-\infty$ and $+\infty$ in recursive calls, attach a two-bit tag distinguishing $-\infty$ padding, genuine records, and $+\infty$ padding; all padding records of one type have the same bit representation. Step~8 removes $r/2$ records from each end of the sorted array, exactly cancelling the $r/2$ copies of each padding type added in step~6. The input's original padding multiplicities are therefore preserved. All recursion levels can share the same two-bit tag without recording the level at which padding was introduced. Any additional $+\infty$ padding appended to reach the required matrix size is likewise removed according to its added multiplicity.

\textbf{Implementing the rearrangements by address-field interchange.} Since $rs=2^{4t}$, each record index $j\in[2^{4t}]$ has $4t$ bits and splits, from highest to lowest, into four fields $(x_1,x_2,x_3,x_4)$ of width $t$. The remainder $j\bmod2^t$ is exactly the lowest field $x_4$, while the quotient $\lfloor j/2^t\rfloor$ is the concatenation $(x_1,x_2,x_3)$ of the highest three fields. Hence step~2 sends $j$ to $\pi(j)=(j\bmod2^t)\cdot2^{3t}+\lfloor j/2^t\rfloor$, moving $x_4$ to the highest position and shifting the other three fields one position to the right:
\[
 (x_1,x_2,x_3,x_4)\longmapsto(x_4,x_1,x_2,x_3).
\]
Move $x_4$ left through adjacent fields to implement this rotation. Numbering field positions from highest to lowest, interchange positions $3,4$, then $2,3$, then $1,2$. This requires three address-field interchanges of width $t$. The inverse permutation $\pi^{-1}$ in step~4 restores $(x_4,x_1,x_2,x_3)$ to $(x_1,x_2,x_3,x_4)$ by applying the same three interchanges in reverse order. In each interchange, a record's own bits form an unchanged suffix coordinate and move with the record. Finally, inserting and deleting padding at the ends of the array in steps~6 and~8 each require just one sequential scan.
\end{proof}

\subsection{The Local Sorting Bound}

Proposition~\ref{prop:columnsort-reduction} reduces sorting the entire array to four rounds of column sorting. If a column still exceeds memory capacity, apply Proposition~\ref{prop:columnsort-reduction} recursively to that column; once a column fits in memory, read it into memory and sort it directly. Each call reduces column length to within a constant factor of the $3/4$ power of the current size. Thus, when the initial record count satisfies $n\le M^2$, a constant number of recursion levels suffices to reach in-memory sorting, yielding the local sorting bound.

\begin{lemma}[Local sorting]
\label{lem:scalar-local}
For $1\le n\le M^2$, an array containing $n$ records can be sorted in
\begin{equation}
 O\!\left(1+\frac nB[1+(\log M/a)^\tau]\right)
 \label{eq:scalar-local}
\end{equation}
block transfers, using at most $M$ words of memory and polynomial disk space.
\end{lemma}
\begin{proof}
Let $H=\lfloor M/C\rfloor$, where the fixed constant $C$ is large enough that memory can hold and sort $512H$ records with two-bit tags while reserving space for the required block buffers and control information. First consider $n\le H^2$.

For a subproblem containing $z$ records, if $z\le512H$, read it into memory, sort it, and write it back to disk. If $z>512H$, choose field width $t=\lceil(\log z)/4\rceil$, which necessarily satisfies $t\ge3$, and set $s=2^t$. Append $+\infty$ padding to reach $rs=2^{4t}$ records, then apply Proposition~\ref{prop:columnsort-reduction}, recursively sorting each column of length $r=s^3=2^{3t}$ in the four column-sorting rounds. At the end of the call, remove the newly added $+\infty$ padding from the suffix.

The choice of $t$ gives padded length $rs<16z$ and column length $r<8z^{3/4}$. Thus, at each recursion level, padding and the four rounds of column sorting increase the total number of records processed by only a constant factor, while each subproblem has size within a constant factor of the $3/4$ power of its parent's size. Starting from $n\le H^2$, after three recursion levels each column has fewer than $128H^{27/32}$ records. Since $H\ge1$, this is strictly below the in-memory sorting threshold $512H$. The recursion therefore has depth at most three, and the total number of records processed over all levels is $O(n)$, including padding and repeated column sorts.

For block transfers, a call of size $z>512H$ that requires further recursion performs six address-field interchanges of width $t=O(\log M)$ on $O(zw)$ bits, together with a constant number of sequential scans. Applying Lemma~\ref{lem:memory-cutoff-interchange} and using $z>512H=\Omega(B)$ to absorb additive constants gives a cost of $O\bigl((z/B)[1+(\log M/a)^\tau]\bigr)$ for the operations at the current level of such a call. At the bottom of the recursion, the short columns of length at most $512H$ in one round are stored consecutively on disk. Reuse the input and output block buffers to read, sort, and write these columns one at a time in a sequential scan. A batch containing $L$ records in total costs only $O(1+L/B)$ block transfers. Summing the costs of at most three recursion levels and the in-memory sorts proves~\eqref{eq:scalar-local}.

Finally, if $H^2<n\le M^2$, then $H=\Theta(M)$ allows us to split the input into a constant number of subarrays of length at most $H^2$. Sort each by the algorithm above and perform a constant number of two-way merge rounds, each costing $O(1+n/B)$ block transfers by sequential scanning. Execute the recursion sequentially in depth-first order and reuse the memory workspace; temporary disk space is polynomial. Corollary~\ref{cor:parameter-transfer} in the appendix extends the result to all $2B\le M$.
\end{proof}

\section{Sorting Arbitrary Input Sizes}
\label{sec:global-sorting}

Section~\ref{sec:sorting} solves local sorting for at most $M^2$ records. When the input size satisfies $N>M^2$, we extend local sorting to arbitrary sizes by radix sorting. Divide the binary representation of each input integer into consecutive bit groups from highest to lowest. In each round, partition the array by the highest remaining bit group into buckets whose records share the same value in that group. Recursively partition each bucket by the next bit group. Once a bucket contains at most $H^2$ records, apply the local sorting algorithm of Section~\ref{sec:sorting} to sort it.

We retain $H=\lfloor M/C\rfloor$ and $a=\log(M/B)$, and write $h=\log M$. Choose bit-group width $t=\lfloor\log H\rfloor=\Theta(h)$; if the highest group has fewer than $t$ bits, use its actual width. Each group has at most $K=2^t\le H$ possible values, so each partition creates at most $K\le H$ buckets. As in the previous two sections, first assume that both $M$ and $M/B$ exceed sufficiently large fixed constants, ensuring $H\ge B$. Corollary~\ref{cor:parameter-transfer} in the appendix extends the resulting bounds to all $2B\le M$.

\subsection{The Cost of One Partition}

Consider partitioning an array of length $n>H^2$ by one bit group. The number of buckets $K$ can be $\Theta(M)$, so memory may not hold a separate buffer of $B$ words for every bucket. We instead use the local sorting algorithm of Section~\ref{sec:sorting} to form the buckets indirectly. Split the current array into subarrays of at most $H^2$ records, and sort each by the current bit group so that records for the same bucket form a contiguous interval within each subarray. Then, in increasing bucket-value order, read the corresponding intervals from the subarrays and concatenate them into complete buckets. There are $O(n/H)$ subarray--bucket pairs. Since $H\ge B$, the boundary transfers incurred when concatenating intervals from different subarrays are absorbed by the $O(n/B)$ scan cost.

\Needspace{7\baselineskip}
\begin{lemma}[Partitioning by one bit group]
\label{lem:digit-partition}
An array of length $n>H^2$ can be partitioned by a bit group of at most $t$ bits into contiguous buckets in increasing bit-group order, using at most
\[
 O\!\left(\frac nB\left[1+\left(\frac ha\right)^\tau\right]\right)
\]
block transfers. After local sorting within the subarrays, concatenating their intervals into contiguous buckets takes only $O(n/B)$ block transfers.
\end{lemma}
\begin{proof}
Following the current storage order, split the array of length $n$ into $J=\lceil n/H^2\rceil$ contiguous subarrays, each containing at most $H^2$ records. Apply Lemma~\ref{lem:scalar-local} to each subarray with the current bit-group value as its sorting key. Since $J=O(n/H^2)$ and $\log M=h$, summing the local sorting costs over all $J$ subarrays gives the stated block-transfer bound.

After a subarray is sorted, the records with any bit-group value $v\in[K]$ occupy one contiguous interval within that subarray; an absent value has an interval of length zero. Scan the $J$ subarrays once, recording each value interval's starting position and length in an auxiliary index table of shape $J\times K$ on disk. Each entry occupies a constant number of machine words. Then process bit-group values $v=0,1,\dots,K-1$ in increasing order. For each value $v$, consult the table, read its intervals from subarrays $1,2,\dots,J$ in turn, and append them sequentially to the output stream. All records with value $v$ thereby form one contiguous bucket on disk.

We now count the block transfers for scanning and concatenation. The auxiliary index table has $JK$ entries. Since $n>H^2$ and $K\le H$, we have $J=\lceil n/H^2\rceil\le2n/H^2$, giving $JK=O(n/H)$. Scanning the subarrays and writing the index entries costs $O(n/B+JK)$ block transfers. During concatenation, each record is read and written sequentially once, costing $O(n/B)$ full-block transfers. Accessing each index entry and handling the partial blocks at each interval's ends incurs only a constant number of additional transfers, for a total overhead of $O(JK)=O(n/H)$. Since $H\ge B$, this overhead is also $O(n/B)$. The entire scan and concatenation phase therefore costs $O(n/B)$ block transfers and requires only a fixed number of block buffers in memory. The resulting buckets are ordered by the current bit-group value; recursion on later bit groups or local sorting determines the order within each bucket.
\end{proof}

\subsection{The Cost of High-Order Recursion}

Lemma~\ref{lem:digit-partition} bounds the cost of one bucket partition by a single bit group. Starting with the highest bit group, we recursively partition buckets that exceed the local sorting threshold until all buckets are sorted. Write $p=\log(N/M)$.

\begin{proposition}[Sorting large inputs]
\label{prop:msd-sorting}
For $N\ge M^2$, the $N$ input integers can be sorted in
\[
 O\!\left(\frac NB\left[1+\frac{p}{a^\tau h^\varepsilon}\right]\right)
\]
block transfers, using at most $M$ words of memory and polynomial disk space.
\end{proposition}
\begin{proof}
Throughout the recursion, all integers in a bucket agree on the high-order bit groups already processed. For a bucket with $n$ records, we distinguish three cases according to its size and the remaining bit groups. If $n\le H^2$, apply Lemma~\ref{lem:scalar-local} to sort the bucket by the complete integers and append it to the final output stream. If $n>H^2$ but all bit groups have been processed, all integers in the bucket are equal, so append them sequentially without further sorting. If $n>H^2$ and some bit groups remain, apply Lemma~\ref{lem:digit-partition} to the highest unprocessed group to form contiguous child buckets, and recursively process the nonempty children by the next bit group in increasing bucket-value order. High-order groups determine the order between buckets, while recursion and local sorting order the records within each bucket. The final output stream is therefore sorted.

We first count the block transfers at one recursion depth. The buckets at that depth are disjoint and contain at most $N$ records altogether. At most $N/H^2$ large buckets have more than $H^2$ records and require further partitioning by Lemma~\ref{lem:digit-partition}. Summing that lemma's partitioning costs over all such buckets gives at most
\begin{equation}
 O\!\left(\frac NB\left[1+\left(\frac ha\right)^\tau\right]\right)
 \label{eq:digit-local-cost}
\end{equation}
block transfers.
Each of these at most $N/H^2$ large buckets creates at most $K\le H$ children, so the number of new child buckets at one depth is at most $(N/H^2)K\le N/H$. Those with at most $H^2$ records are sorted directly by Lemma~\ref{lem:scalar-local}. They contain at most $N$ records altogether, and each of the at most $N/H$ short buckets incurs an $O(1)$ startup cost. Since $H\ge B$, their combined startup cost $O(N/H)$ is absorbed by $O(N/B)$. Their total local sorting cost is thus also bounded by~\eqref{eq:digit-local-cost}. Each new bucket's task record incurs only a constant number of block transfers, so the same bound absorbs task management. Large buckets with no remaining bit groups contain equal integers and can be output sequentially for a total cost of $O(N/B)$. Hence each recursion level, together with sorting the short buckets it creates, satisfies~\eqref{eq:digit-local-cost}.

We now sum over the recursion tree. Each integer has $w=\Theta(\log N)$ bits. Except for the highest group, which may be shorter, each partitioning level processes $t=\Theta(h)$ bits. The number of levels is therefore $O(1+(\log N)/h)$. When $N\ge M^2$, we have $p=\log(N/M)\ge\log M=h$ and $\log N=p+h\le2p$, so the depth is $O(p/h)$. Multiplying the per-level cost in~\eqref{eq:digit-local-cost} by this depth, and using $\tau=1-\varepsilon$ and $a=\log(M/B)\le\log M=h$, gives
\[
 O\!\left(\frac NB\left[\frac ph+\frac ph\left(\frac ha\right)^{1-\varepsilon}\right]\right)
 =O\!\left(\frac NB\left[1+\frac{p}{a^\tau h^\varepsilon}\right]\right).
\]
Here the first term $p/h$ is absorbed by the second term $p/(a^\tau h^\varepsilon)$ because $a\le h$.

For space management, process buckets sequentially in depth-first order. After partitioning a large bucket, push each nonempty child's starting position, length, and next bit-group index onto a disk task stack, in decreasing bucket-value order so that children are popped in increasing order. Each task record occupies a constant number of machine words. Recursion management keeps only the current task and a fixed number of block buffers in memory; partitioning and local sorting reuse their memory workspace as needed. The depth is $O(\log N)$ and each level has at most $N$ active records, so temporary disk space is polynomial. Finally, Corollary~\ref{cor:parameter-transfer} in the appendix extends the result to all $2B\le M$.
\end{proof}

\subsection{Sorting with Bounded Memory}
\label{sec:direct-radix}

For $M$ below the fixed memory threshold, direct radix sorting gives a linear I/O bound.

\begin{lemma}[A linear number of individual record accesses]
\label{lem:direct-radix}
For every $M\ge2B$, an array of $N$ integers of $w=\Theta(\log N)$ bits can be sorted in $O(N)$ block transfers, using at most $M$ words of memory and polynomial disk space.
\end{lemma}
\begin{proof}
First suppose $w\ge\lceil\log N\rceil$. Use radix sorting from the lowest bit group to the highest, processing $\lfloor\log N\rfloor$ bits per pass, so only a constant number of passes is needed. Store a list head for each bucket and a successor pointer for each record on disk. Visit input positions from last to first, inserting each record at the head of its bucket. Then process buckets in increasing value order, following their lists and copying the complete records. Reverse insertion preserves the original input order within each bucket, so every pass is stable and processing all groups produces sorted output. Initializing the heads, inserting records, and reading them out require $O(N)$ individual accesses altogether.

If $w<\lceil\log N\rceil$, there are only $K=2^w<N$ possible integer values, so use counting sort directly. Since $w=\Theta(\log N)$, each counter occupies a fixed number of machine words. Initialize the $K$ counters on disk, scan the input, and update the appropriate counter word by word. Then process the values in increasing order, outputting each integer as many times as its count specifies. Initialization, counting, and output require $O(N+K)=O(N)$ word accesses and preserve every integer's multiplicity.

Store each individually accessed word in its own physical block; conversion to or from packed input/output takes $O(N)$ transfers. Two block buffers suffice: one holds the operand and the other accesses disk. Remaining counters, pointers, and loop state reside on disk; their wordwise accesses are included above. Memory use is at most $2B\le M$ words. For $w<\lceil\log N\rceil$ and $M$ below the fixed threshold, an index records the current step of the predetermined operation sequence and determines input/output positions. Appendix~\ref{app:formalization} places this index outside the $M$ words. The counters and lists occupy $O(N)$ words, giving $O(NB)=O(N^2)$ words of disk space.
\end{proof}

\subsection{Proof of the Main Sorting Theorem}
\label{sec:main-sorting-proof}

We now combine local sorting with radix sorting to prove Theorem~\ref{thm:sorting}.

\begin{proof}[Proof of Theorem~\ref{thm:sorting}]
Write $p=\log(N/M)$, $a=\log(M/B)$, $h=\log M$, and $\tau=1-\varepsilon$. First assume that both $M$ and $M/B$ exceed sufficiently large fixed constants, and distinguish two input-size ranges.

The first case is $M\le N\le M^2$. Apply Lemma~\ref{lem:scalar-local} directly, including the boundary $N=M$; its memory allocation already reserves buffers and control workspace. Since $N\ge M\ge2B$, the scan term $N/B$ absorbs the constant startup cost. Also, $h=\log M\le\log N$, so the lemma's cost factor satisfies
\[
 1+\left(\frac ha\right)^\tau
 =1+\frac{h}{a^\tau h^\varepsilon}
 \le1+\frac{\log N}{a^\tau h^\varepsilon}.
\]
Thus the block-transfer cost satisfies~\eqref{eq:simple-sorting} throughout $M\le N\le M^2$.

The second case is $N>M^2$. Apply Proposition~\ref{prop:msd-sorting} directly, at a cost of
\[
 O\!\left(\frac NB\left[1+\frac{p}{a^\tau h^\varepsilon}\right]\right)
\]
block transfers. Since $p=\log(N/M)\le\log N$, this cost also satisfies~\eqref{eq:simple-sorting}.

All these algorithms are deterministic, correct on every input, and use at most $M$ words of memory and polynomial disk space. Corollary~\ref{cor:parameter-transfer} in the appendix extends the result to every memory-to-block-size ratio satisfying $2B\le M$. If $M$ is below the fixed threshold, then $B$ is bounded, and the $O(N)$ block transfers of Lemma~\ref{lem:direct-radix} are also $O(N/B)$. Taking the fixed constant $\varepsilon=1.5\times10^{-5}>0$ and combining the cases proves the theorem.
\end{proof}

\section{Conclusion}
\label{sec:conclusion}

This paper studies integer sorting in the external-memory model, focusing on reducing block transfers between disk and memory. We prove that $N$ integers can be sorted in
\[
 O\!\left(\frac NB\left[1+\frac{\log N}{(\log(M/B))^{1-\varepsilon}(\log M)^\varepsilon}\right]\right)
\]
block transfers. When $N\ge M^2$ and $\log M/\log(M/B)\to\infty$, this bound asymptotically improves on multiway mergesort, capturing the advantage when blocks are large and the number $M/B$ of memory blocks is relatively small. The algorithm combines three components. At the lowest level, a coded transfer scheme implements binary address-field interchange. At the middle level, recursive Columnsort reduces local sorting of at most $M^2$ records to column sorting and address-field interchange. At the outer level, radix sorting recursively partitions by bit groups from highest to lowest, extending local sorting to arbitrary input sizes.

Designing a coded transfer scheme for sorting as a whole is difficult. Our approach shows that combining radix sorting with Columnsort reduces the overall problem to local sorting and binary address-field interchange. Designing a code for this basic interchange operation then suffices to reduce block transfers throughout the sorting algorithm. Using the interchange construction of~\cite{OpenAITapes2026}, Appendix~\ref{sec:improved-interchange} extends the sorting bound to every fixed $0<\varepsilon<1/2$. Whether similar coding techniques can accelerate other external-memory algorithms remains open.

\FloatBarrier
\appendix
\section{Buffered I/O and Memory Simulation}
\label{app:io-simulation}

\subsection{Simulating sequential tape access}
\label{app:tape-access}

To implement the invertible lower-triangular coordinate transformations $L$ and $L^{-1}$ in Section~\ref{sec:coding-recursion}, we use the rearrangement routines of OpenAI~\cite[Lemmas~3--4]{OpenAI2026}. These routines were originally described in the multitape Turing machine model. We first explain how such a machine accesses its tapes, and then show how block buffers in memory simulate these operations in our external-memory model.

A multitape Turing machine stores its data on a fixed number of tapes. Each tape is a sequence of cells arranged in order, with a head that reads and writes them. In one elementary tape step, each head reads or changes the contents of its current cell, and then moves one cell to the left, one cell to the right, or stays in place. In the simplest binary model, each cell stores one bit. For convenience, an algorithm may instead store a symbol from a fixed finite alphabet in each cell, or use a fixed number of parallel tracks on a tape so that its head can read and write a constant number of bits at the same position in one step. In all these forms, each cell occupies a fixed constant number of bits. A disk block containing $b=Bw$ bits can therefore hold $\Theta(b)$ consecutive tape cells.

To simulate these operations in external memory, we store each tape as a consecutive sequence of bits on disk. Accessing the disk at every tape step would make the number of block transfers grow linearly with the number of steps $T$, which is too expensive. Instead, we keep a sliding window of three consecutive disk blocks in memory for each tape. The head can move within its window entirely in memory, and must move at least $\Theta(b)$ steps before the window needs updating.

\begin{lemma}[Buffered tape simulation]
\label{lem:buffered-tape}
Suppose a multitape Turing machine process uses a fixed number of tapes and performs $T$ elementary tape steps. It can be simulated in the external-memory model using $O(B)$ words of memory and $O(1+T/(Bw))$ block transfers. More generally, suppose the computation consists of $J$ successive phases with $T$ elementary tape steps in total, and each phase begins by loading the tape windows into memory and ends by writing modified blocks back to disk. Simulating all $J$ phases then takes $O(J+T/(Bw))$ block transfers.
\end{lemma}
\begin{proof}
For each tape, we maintain a sliding window of three consecutive disk blocks in memory. Whenever a window is loaded, it contains the block holding the head and the two blocks immediately to its left and right. Since the number of tapes is fixed and each disk block occupies $B$ machine words, all the windows together use $O(B)$ words of memory.

At the start of a phase, we load three consecutive blocks for each tape, centered on the block containing its initial head position. As long as a head stays within its current three blocks, all reads, writes, and head movements on that tape take place in the memory window and incur no disk transfers. When a head leaves its window, we write the modified blocks back to disk. We then load a new three-block window centered on the block the head has just reached, including that block and its two neighbors.

The three-block window prevents frequent disk accesses when a head moves back and forth across a block boundary. After each window update, the head lies in the middle block of the new window. To leave this window again in either direction, it must cross an entire neighboring block, requiring at least $\Theta(b)=\Theta(Bw)$ elementary tape steps. Thus, if the head on tape $i$ makes $T_i$ moves during a phase, it triggers at most $O(T_i/b)$ window updates. Each update takes a constant number of block transfers, as do the initial load and final writeback. Tape $i$ therefore incurs $O(1+T_i/b)$ transfers during that phase. Summing over the fixed number of tapes gives $O(1+T/(Bw))$ transfers for one phase; summing over all $J$ phases gives $O(J+T/(Bw))$ transfers in total.
\end{proof}

Lemma~\ref{lem:buffered-tape} applies directly to the stream operations within one level of the coding recursion in Section~\ref{sec:coding-recursion}. Splitting, copying, bitwise XOR, and merging are sequential reads and writes of disk streams. The lower-triangular transformations $L$ and $L^{-1}$ use the rearrangement routines of OpenAI~\cite[Lemmas~3--4]{OpenAI2026}, following the construction in Section~\ref{sec:coding-recursion}. Each data stream is stored consecutively in lexicographic address order. The $2m$ short fields being transformed are coordinates over $\mathbb{Z}/q^f\mathbb{Z}$, and the preselected odd prime $q$ ensures that $L$ and $L^{-1}$ are invertible over this ring; all other coordinates remain unchanged. This input format and these algebraic conditions satisfy the requirements of the cited routines. Since the dimension and coefficients of the matrices are fixed, the transformations use a fixed number of tapes. Successive prefix segments reuse the work tapes, and the total number of tape steps is linear in the number of bits in the stream. The circuit has fixed size, so all stream operations at one recursion level, apart from recursive calls, form $J=O(1)$ tape phases with $T=O(V)$ steps in total. Lemma~\ref{lem:buffered-tape} therefore gives $O(1+V/(Bw))$ block transfers. Finally, array dimensions and address offsets are bounded by a fixed polynomial in $N$ and occupy a constant number of $w$-bit machine words. The associated calculations take place in memory and incur no block transfers.

\subsection{Suspending and restoring recursive calls}
\label{app:recursive-suspension}

In the recursion of Section~\ref{sec:coding-recursion}, each call generates $S$ child calls, executed one after another. If every recursion level kept its own block buffers in memory, memory use would grow with the recursion depth. To use only a fixed number of memory blocks, we follow the depth-first execution arrangement of OpenAI~\cite[Proposition~11]{OpenAI2026}, running only one call at a time.

\textbf{Suspending the parent and freeing memory.} Before entering a child call, write the parent's buffered blocks back to disk. Also save its current parameters on disk: the four array dimensions $(P,G,\ell,e)$, the disk locations of its data streams, the tape-head positions, the current execution step, and which child call is being processed. This releases the memory occupied by the parent, allowing the child to use the same space for its own buffers and variables.

\textbf{Finishing the child and restoring the parent.} After the child has processed its $V/W$ bits, write its buffered data back to disk and copy the result of the field interchange into the corresponding parent stream. Then reload the parent's saved parameters and buffered blocks, and continue with its next step. The next child reuses the temporary disk space used by the previous one.

\textbf{Block transfers and space.} The parent's parameters occupy a fixed number of machine words. Saving or restoring them and switching block buffers therefore takes $O(1)$ block transfers. Copying $V/W$ bits for one child takes $O(1+V/(WBw))$ transfers, and initializing its workspace requires only the corresponding sequential scans. Since the circuit is fixed, both $S$ and $W$ are constants, so all $S$ children together add only $O(1+V/(Bw))$ transfers. Adding the stream-operation cost from Appendix~\ref{app:tape-access} gives the recurrence in Lemma~\ref{lem:coded-recursion}. Before entering the next level, the previous call's data are written back to disk. Memory therefore holds only the current call's buffers, one block buffer for reading and writing the saved parameters, and a constant number of control words, using $c_1B+c_2$ words for fixed constants $c_1,c_2$. Appendix~\ref{app:block-simulation} explains how to reserve this space. The number of unfinished calls stored on disk is at most the recursion depth $O(\log(2+e))$. No level has more data than the padded initial array, so total disk space is polynomial in $N$, and every disk address occupies only a constant number of $w$-bit machine words.

\subsection{Constant-factor block simulation}
\label{app:block-simulation}

The algorithms in Sections~\ref{sec:interchange}--\ref{sec:global-sorting} keep a fixed number of block buffers in memory. The main text therefore first assumes $M/B\ge R_0$ and $M\ge C_0$, for sufficiently large fixed constants $R_0,C_0$. When memory holds only a few physical blocks, as for $M=2B$, it cannot hold $R_0$ buffers of $B$ words each. We reduce the algorithm's working memory $\widehat M$ and logical block size $\widehat B$ by fixed factors so that $\widehat M/\widehat B\ge R_0$. Disk transfers still read or write an entire physical block of $B$ words. We therefore reserve one $B$-word transfer buffer, receive the whole physical block there, and pass its first $\widehat B$ valid words to the algorithm.

\begin{lemma}[Constant-factor block simulation]
\label{lem:io-normalization}
For any fixed constants $R_0\ge2$ and $C_0$, there are fixed constants $K,C_*$ such that, whenever $M\ge\max\{2B,C_*\}$, setting
\[
 \widehat M=\lfloor M/4\rfloor,
 \qquad \widehat B=\max\{1,\lfloor B/K\rfloor\},
\]
gives $\widehat M\ge C_0$ and $\widehat M/\widehat B\ge R_0$. An algorithm using $\widehat M$ words and logical blocks of $\widehat B$ words can be simulated in $M$ words, with one physical transfer per logical block read or write. For input and output totaling $L$ words, format conversion costs $O(1+L/B)$ transfers, and disk space grows by only a fixed factor. Moreover, $\log(\widehat M/\widehat B)=\Theta(\log(M/B))$ and $\log\widehat M=\Theta(\log M)$.
\end{lemma}
\begin{proof}
\textbf{Parameter relations.} Choose a fixed integer $K\ge4R_0$ and then $C_*$ large enough that $\widehat M\ge M/8\ge\max\{R_0,C_0\}$. If $B\ge K$, then $M\ge2B$ gives $\widehat M/\widehat B\ge KM/(8B)\ge R_0$. If $B<K$, then $\widehat B=1$ and again $\widehat M/\widehat B\ge R_0$. Also $B/(2K)\le\widehat B\le B$. Thus the memory and block sizes change by only fixed factors. Both memory-to-block ratios are at least two, so the stated logarithmic relations follow.

\textbf{Transfers and format conversion.} Store each logical block in the first $\widehat B$ words of its own physical block, leaving the rest unused. Allocate at most $M/4$ words to the algorithm, one $B\le M/2$ word buffer to disk transfers, and the remaining memory to control variables. To read a logical block, read its physical block into the transfer buffer and copy the first $\widehat B$ words into working memory. To write a logical block, place its $\widehat B$ words in the buffer and write the whole physical block. Each logical block read or write therefore takes one physical transfer.

Before running the algorithm, extract groups of $\widehat B$ input words in their original order and write each group as a logical block. A group spans at most two original physical blocks. Read these blocks one at a time, save the needed words in working memory, and reuse the transfer buffer to write the group. After the algorithm finishes, extract the valid words from its logical blocks in the prescribed output order and pack them into physical blocks. Temporarily hold one logical block in working memory, then read each target physical block, fill the appropriate positions, and write it back. A group spans at most two target blocks. Both conversions take a constant number of transfers per logical block, totaling $O(1+L/\widehat B)=O(1+L/B)$. Storage per logical block increases by at most $B/\widehat B\le2K$, so total disk space still increases by only a fixed factor.
\end{proof}

\begin{corollary}[Extension to general parameters]
\label{cor:parameter-transfer}
The address-field interchange bound of Lemma~\ref{lem:memory-cutoff-interchange} holds for every $M\ge\max\{2B,C_{\rm abs}\}$. The local sorting bound of Lemma~\ref{lem:scalar-local}, the large-input sorting bound of Proposition~\ref{prop:msd-sorting}, and the sorting bound of Theorem~\ref{thm:sorting} hold for every $2B\le M$ under their respective input-size conditions.
\end{corollary}
\begin{proof}
\textbf{Address-field interchange and sorting with larger memory.} For $M\ge C_*$, take $C_{\rm abs}\ge C_*$. Use Lemma~\ref{lem:io-normalization} to convert the input to logical blocks, run the main-text interchange or sorting algorithm at parameters $\widehat M,\widehat B$, and pack the output back into physical blocks. The block size and the logarithms of memory size and block count change by only constant factors. Thus the interchange and local sorting bounds change by only constant factors, and their bounds absorb the linear conversion cost. For Proposition~\ref{prop:msd-sorting}, $N\ge M^2$ implies $N\ge\widehat M^2$ and $\log(N/M)\ge\log M$. Reducing memory adds only a constant to $\log(N/M)$, so the large-input sorting bound also changes by only a constant factor. Combining the input-size cases of Section~\ref{sec:main-sorting-proof} gives the bound of Theorem~\ref{thm:sorting}.

For local sorting, also consider $\widehat M^2<n\le M^2$. Since $\widehat M=\Theta(M)$, split the array into a constant number of subarrays of length at most $\widehat M^2$. Sort them at the simulated parameters, then perform a constant number of two-way sequential merge rounds at the same simulated parameters. The original bound absorbs the merge and conversion scans.

\textbf{Sorting with bounded memory.} If $2B\le M<C_*$, both $M$ and $B$ are bounded by the fixed constant $C_*$. A local input of at most $M^2$ records has constant size and needs only $O(1)$ transfers. For a general input of size $N$, apply Lemma~\ref{lem:direct-radix}: its $O(N)$ transfers equal $O(N/B)$ for bounded $B$, satisfying all the stated sorting bounds. These simulations, merges, and direct sorts preserve correctness, the $M$-word memory bound, and polynomial disk space.
\end{proof}

\section{Verification of the Fixed Parameters}
\label{app:parameters}

This appendix verifies the inequality $\varepsilon\ln m<\eta$ used in Corollary~\ref{prop:finite-circuit} of Section~\ref{sec:network}.

\begin{lemma}[Parameter inequality]
\label{lem:exact-exponent}
For $m=512$, $\varepsilon=3/200000=1.5\times10^{-5}$, and $\eta=33/338432$, we have $\varepsilon\ln m<\eta$.
\end{lemma}
\begin{proof}
The first four terms of the exponential series give
\[
 e^{7/10}>1+\frac7{10}+\frac{(7/10)^2}{2}
                +\frac{(7/10)^3}{6}>2.
\]
Thus $\ln2<7/10$, so $\ln512=9\ln2<63/10$. Therefore,
\[
 \varepsilon\ln512<\frac{3}{200000}\times\frac{63}{10}
 <\frac{33}{338432}=\eta.\qedhere
\]
\end{proof}

\section{Arbitrary Field Widths and Prefix Counts}
\label{app:arbitrary-fields}

Lemma~\ref{lem:interchange-scales} in Section~\ref{sec:recursion-bound} requires a field width of the form $e=t_0m^k$ and a prefix count $P$ divisible by $W^k$. For general $e\ge1$ and $P\ge1$, first handle calls needing no recursion and fields that are too short. In the remaining cases, let $k$ be the required recursion depth. Move the highest $\rho$ digits of each field $H,D$ into the prefix, dividing each original segment into $R=q^{2\rho}\ge W^k$ genuine smaller segments. Append fewer than $W^k$ zero segments to each original segment so that the count becomes a multiple of $W^k$, increasing the data volume by less than a factor of two. Then use the base-$m$ expansion of the remaining width to split the low fields into pieces of widths $t_0m^j$. Apply Lemma~\ref{lem:interchange-scales} to the entire padded array for each pair of pieces, completing the interchange within each smaller segment; any residual piece shorter than $t_0$ is interchanged by scans. Each completed interchange fixes the segment prefixes. Finally, remove the trailing zero segments from each original segment and move the high digits to their positions in the permuted layout.

\begin{lemma}[Interchange for arbitrary field widths and prefix counts]
\label{lem:arbitrary-fields}
Use the fixed parameters of Lemma~\ref{lem:interchange-scales}, with $t_0=\Theta(\log(M/B))$. For arbitrary positive integers $P,G,\ell,e$, consider a bit array stored on disk in lexicographic order with addresses
\[
 (p,H,g,D,z)\in[P]\times[q^e]\times[G]\times[q^e]\times[\ell],
\]
where $P,G,\ell,q^e$ are bounded by a fixed polynomial in $N$. Interchanging the address fields by $(p,H,g,D,z)\mapsto(p,D,g,H,z)$ on all $V=PG\ell q^{2e}$ bits takes at most
\begin{equation}
 Q(e,V)=O\!\left(1+\frac{V}{Bw}\left[1+\left(\frac{e}{t_0}\right)^\tau\right]\right)
 \label{eq:cutoff-general-q}
\end{equation}
block transfers. Padding increases the data volume by less than a factor of two. The algorithm uses at most $M$ words of memory and polynomial disk space.
\end{lemma}
\begin{proof}
\textbf{Recursion depth and short fields.} If $V\le Mw/2$, read the array into memory and perform the interchange there. If $e\le t$, use the sequential-scan method of Section~\ref{sec:recursion-bound}. Otherwise, put $x=e/t_0>1$, choose maximum recursion depth $k=\lceil\log_m x\rceil$, and take the smallest positive integer $\rho=\lceil k\log W/(2\log q)\rceil$ for which $R=q^{2\rho}\ge W^k$. If $\rho\ge e$, then $t_0\ge1$ and $\rho=\Theta(k)$ give $x=e/t_0\le e\le\rho=O(1+\log x)$. Since linear growth exceeds logarithmic growth, this inequality bounds $x$ by a fixed constant. Split the fields into a constant number of pairs of width at most $t_0$ and interchange them by sequential scans, at a total cost of $O(1+V/(Bw))$. Henceforth assume $\rho<e$.

\textbf{Moving high digits into the prefix and interchanging the fields in three steps.} Split each field into a high part of width $\rho$ and a low part of width $h=e-\rho$, writing $H=(H_+,H_-)$ and $D=(D_+,D_-)$. Perform the interchange in three steps:
\begin{equation}
\begin{aligned}
 &(p,H_+,H_-,g,D_+,D_-,z)\\
 \xmapsto{E}\;&(p,D_+,H_+,H_-,g,D_-,z)\\
 \xmapsto{L}\;&(p,D_+,H_+,D_-,g,H_-,z)\\
 \xmapsto{F}\;&(p,D_+,D_-,g,H_+,H_-,z).
\end{aligned}
\label{eq:arbitrary-fields-layouts}
\end{equation}
Step $E$ moves $D_+$ forward to just after $p$ within each prefix segment. Step $L$ fixes $(p,D_+,H_+)$ and interchanges $H_-,D_-$. Step $F$ moves $H_+$ backward to just before $H_-$ within each $(p,D_+)$ segment. Steps $E$ and $F$ each move $\rho$ base-$q$ digits, preserving their order within each high part.

To move one digit, use $q$ stream buffers and one output buffer. First scan the whole array, distributing bits among $q$ streams according to the digit and preserving the order of all other coordinates. Then read the required data from these streams in lexicographic target-address order. For a forward move, process each outer prefix by concatenating its portions of streams $0$ through $q-1$. For a backward move, select the stream by the digit's value in the target address and read the required consecutive portion. Once that digit is fixed, all other coordinates have the same lexicographic order before and after the move, so every stream is read sequentially. Different outer prefixes continuously share the streams and buffers; a prefix boundary does not flush a partial block. Remaining partial blocks are handled only at the end of each complete stream. Since $q$ is fixed, one digit move takes a constant number of sequential reads and writes over the whole array, costing $O(1+V/(Bw))$ transfers. Moving $\rho$ digits therefore takes $O(\rho)$ scans.

\textbf{Padding the segments and splitting the low fields.} After step $E$, the $R=q^{2\rho}$ values of $(D_+,H_+)$ divide each original segment $p\in[P]$ into $R$ consecutive genuine segments. Let $R'$ be the smallest multiple of $W^k$ at least $R$. Since $R\ge W^k$, we have $R'<2R$. In one sequential scan, append $R'-R$ equally long zero segments after the $R$ genuine segments for each $p$. Each segment has $q^{2h}G\ell$ bits, the count $PR'$ is divisible by $W^k$, and the total volume is $V'=(R'/R)V<2V$.

Next perform step $L$ on these $PR'$ segments, interchanging the low fields of width $h=e-\rho$. Divide $h$ by $t_0$ and expand the quotient in base $m$ as $\sum_{j=0}^k a_jm^j$, with $0\le a_j<m$ and remainder $0\le r<t_0$. Split each low field into corresponding consecutive pieces: $a_j$ pieces of width $t_0m^j$ for each $j$, and a final piece of width $r$ if $r>0$. Interchange corresponding pairs from left to right. For a pair of width $t_0m^j$, the preceding low-field pieces join the prefix. If their total width is $b$, the current prefix count is $PR'q^b$, which is divisible by $W^j$. Apply Lemma~\ref{lem:interchange-scales} to the entire padded array for this pair; interchange the residual short pair by sequential scans. Earlier interchanges preserve every coordinate's number of possible values, so the divisibility condition holds for every call. Each interchange fixes all other coordinates, and interchanging every pair implements $L$.

Measured in units of the cost of one sequential scan of the whole array, the calls have total cost
\[
 O\!\left(1+\sum_{j=0}^k a_jm^{j\tau}\right)
   =O(m^{k\tau})=O(x^\tau).
\]
Since $m^{k\tau}\le m^\tau x^\tau$ and $\tau>0$, this also absorbs the $O(\rho)=O(1+\log x)$ scans for moving the high digits. Moreover, $V>Mw/2$ and $M\ge2B$ give $V/(Bw)>1$, so each call's constant number of additional transfers is included in this scan cost.

Every completed low-field piece interchange fixes the segment prefix and preserves the bit values. Thus, after all interchanges, the $R'-R$ trailing zero segments for each $p$ remain in place. Remove them in a sequential scan, then apply step $F$ to move $H_+$ to its position in the permuted layout. This completes the required interchange.

All positions and segment lengths are determined by the array dimensions and can be computed successively using a constant number of machine words. Scans, piece interchanges, and recursive calls run one at a time and reuse workspace. Together with the buffering and memory simulation of Appendix~\ref{app:io-simulation}, this keeps memory use within $M$ words. Padding increases the volume by only a constant factor, and the recursion depth is $O(\log(2+e))$, so disk space remains polynomial in $N$.
\end{proof}

\section{The Execution-Control Convention}
\label{app:formalization}

This section explains how the program's execution position and working storage are accounted for, and how recursive control state is stored.

\textbf{Accounting for the execution position.} Aggarwal and Vitter~\cite[Sections 2 and 5]{AV1988} define $M$ as the number of data records that fit in memory. For example, their FFT algorithm loads $M$ records into memory at once, without reserving a record location for the program's execution position. Vitter's survey~\cite[Sections 2.1 and 5.2]{Vitter2001} and the sorting paper of Arge and Thorup~\cite[Sections 1 and 2]{ArgeThorup2015} likewise describe memory usage in terms of the data and buffers held in memory. We specify below how our model accounts for the program's execution position.

\textbf{Storing recursive control state.} The recursion in Section~\ref{sec:coding-recursion} follows the depth-first execution in Appendix~\ref{app:recursive-suspension}. Memory holds only a constant number of control variables for the current call: the array dimensions, the current stage, the child-call index, and the locations of the data streams. Before entering a child call, the parent variables are written to disk; when the child returns, they are read back and execution resumes. Each save and restore costs a constant number of block transfers, and the amount of control information stored on disk is proportional to the recursion depth. These costs are already included in Appendix~\ref{app:recursive-suspension}, while Appendix~\ref{app:block-simulation} reserves memory for the current control variables and block buffers. Thus deeper recursion does not increase memory use.

\textbf{Execution with small memory.} When $M$ is below a fixed threshold, sorting uses the direct algorithm of Lemma~\ref{lem:direct-radix}. If also $w<\lceil\log_2N\rceil$, the construction stores its predetermined execution-step index outside the $M$ words. The index records the position in an operation sequence fixed before the input values are known. Apart from this index and program code, data buffers, counters, pointers, and other mutable workspace count toward the memory and disk space as specified in Section~\ref{sec:results}. The threshold depends only on the construction and the constants in $w=\Theta(\log N)$. Under this convention, the block-transfer and space bounds of Theorem~\ref{thm:sorting} hold.

\section{Improved Sorting Exponent via Address-Field Interchange}
\label{sec:improved-interchange}

In Sections~5--8, the exponent $\varepsilon = 1.5 \times 10^{-5}$ in Theorem~\ref{thm:sorting} arises solely from the recursive volume multiplier of the explicit circuit in Section~5, where each recursion level shrinks the address-field width by a factor of $m$ while expanding the total data volume to at most $\lambda V$ with $\lambda < m^{1-\varepsilon}$. Both Columnsort in Section~7 and high-to-low radix partitioning in Section~8 invoke Lemma~\ref{lem:memory-cutoff-interchange} as a black box, without relying on the internal structure of the circuit in Sections~5 and~6. Consequently, any external-memory address-field interchange algorithm that improves the exponent $\tau = 1 - 1.5 \times 10^{-5}$ in Lemma~\ref{lem:memory-cutoff-interchange} to an arbitrary fixed $\alpha \in (1/2, 1)$ immediately extends Theorem~\ref{thm:sorting} to any fixed $\varepsilon = 1 - \alpha \in (0, 1/2)$ without modifying the sorting reduction.

In subsequent work on multitape integer multiplication, OpenAI~\cite[Theorem~2.1 and Section~2.5]{OpenAITapes2026} constructed a recursive address-field interchange scheme in which, for any $\alpha \in (1/2, 1)$, each recursion level divides the field width by at least a constant $m > 1$, keeps the aggregate volume of all subcalls within $\lambda V < m^{\alpha}V$, and performs only linear sequential tape work at each node. Because this construction operates via sequential accesses on a constant number of tapes, we can execute it directly in external memory using the sliding-window simulation of Appendix~\ref{app:tape-access} and the depth-first suspension of Appendix~\ref{app:recursive-suspension}. Whereas the multitape algorithm recurses down to fields of bounded width, external memory can hold $M/B$ blocks simultaneously; truncating the recursion once the field width reaches $\Theta(\log(M/B))$ allows the base case to finish via buffered sequential scans, bounding the recursion depth by $O_{\alpha}(1+\max\{0,\log_m(u/\log(M/B))\})$ and yielding a scan-cost multiplier of $O_{\alpha}(1+(u/\log(M/B))^{\alpha})$. Corollary~\ref{cor:near-half-sorting} gives the resulting upper bound.

\begin{corollary}
\label{cor:near-half-sorting}
In the external-memory model of Theorem~\ref{thm:sorting}, for every fixed $0<\varepsilon<1/2$ and every $2B\le M\le N$, $N$ integers of $\Theta(\log N)$ bits can be sorted deterministically in
\[
  O_{\varepsilon}\!\left(\frac{N}{B}\left[1 + \frac{\log N}{(\log(M/B))^{1-\varepsilon}(\log M)^{\varepsilon}}\right]\right)
\]
block transfers, using at most $M$ words of memory and polynomial disk space. In particular, when $M = 2B$ and $N = B^2$, sorting takes $O_{\delta}\!\left(B(\log B)^{1/2+\delta}\right)$ block transfers for every fixed $0 < \delta < 1/2$.
\end{corollary}

\begin{proof}
Fix $0 < \varepsilon < 1/2$ and set $\alpha = 1 - \varepsilon \in (1/2, 1)$, $a = \log(M/B)$, and $h = \log M$. As in Sections~6--8, we first assume that $M$ and $M/B$ exceed sufficiently large constants depending on $\alpha$; the remaining parameter ranges follow from Corollary~\ref{cor:parameter-transfer} and Lemma~\ref{lem:direct-radix}.

\paragraph{External recursive structure.}
By~\cite[Equation~(2.5) and Section~2.5]{OpenAITapes2026}, for every fixed $\alpha \in (1/2, 1)$ there exist a fixed tape alphabet, a radix $q_0 \ge 2$ (written $q_0$ to avoid conflict with the block bit-size $b = Bw$ in Appendix~\ref{app:tape-access}), a width-reduction factor $m > 1$, and a volume multiplier $1 < \lambda < m^{\alpha}$, all depending only on $\alpha$, with the following properties. To interchange two $q_0$-ary address fields of width $e$ in an array of $V$ symbols on a multitape Turing machine, a recursive node performs $O_{\alpha}(V)$ elementary operations on a constant number of tapes and issues $O_{\alpha}(1)$ recursive subcalls of the same form, each on fields of width at most $e/m$, such that the total symbol volume across all subcalls is at most $\lambda V$. In that construction, each address field has a group of $r$ subfields immediately before its $m$ main subfields and another group of $r$ immediately after them. For $1/2<\gamma<\alpha$, the main subcalls process at most $2m^\gamma V$ symbols in total. Interchanging the $2r$ corresponding pairs in the flanking groups adds one subcall of volume $V$ per pair, for a total of $2rV$. The parameters in~\cite[Equation~(2.5)]{OpenAITapes2026} give $\lambda:=2m^\gamma+2r<m^\alpha$. All recursion levels share the same fixed alphabet.

\paragraph{External-memory simulation.}
Following Appendices~\ref{app:tape-access} and~\ref{app:recursive-suspension}, each tape is stored contiguously on disk and subcalls are executed depth-first, writing the parent's control state and sliding windows to disk before descending and restoring them upon return. Since the alphabet is fixed, each symbol occupies $O_{\alpha}(1)$ bits and a single disk block stores $\Theta_{\alpha}(Bw)$ symbols. By Lemma~\ref{lem:buffered-tape}, the $O_{\alpha}(V)$ non-recursive tape steps and suspension/resumption work at a single node require $O_{\alpha}(B)$ words of memory and $O_{\alpha}(1 + V/(Bw))$ block transfers.

\paragraph{Recursion cutoff.}
Following Section~\ref{sec:recursion-bound}, set the width threshold to
\[
  t = \left\lfloor \frac{1}{8}\log_{q_0}(M/B) \right\rfloor = \Theta_{\alpha}(a).
\]
Recursion stops at a node whenever either the entire array fits inside the reserved internal memory, in which case it is read in, rearranged in memory, and written back, or $e \le t$. In the latter case, the number of distinct value pairs for the two $q_0$-ary fields is at most $q_0^{2e} \le (M/B)^{1/4}$, so allocating one $B$-word block buffer per pair completes the interchange in a constant number of sequential scans, as in Case~1 of Lemma~\ref{lem:interchange-scales}. Both base cases use $O_{\alpha}(1 + V/(Bw))$ block transfers.

\paragraph{Block-transfer bound.}
Every internal node that continues recursing cannot fit in the reserved memory and hence satisfies $V = \Omega_{\alpha}(Mw)$. For sufficiently large $M/B$, this gives $V/(Bw) = \Omega_{\alpha}(M/B) = \Omega_{\alpha}(1)$, absorbing the $O_{\alpha}(1)$ startup and switching overheads of the node and its $O_{\alpha}(1)$ children into its $O_{\alpha}(V/(Bw))$ scan cost. Starting from width $e$ and volume $V$, the recursion depth is at most $k = \max\{0, \lceil \log_m(e/t) \rceil\}$, and at depth $j \in \{0, \dots, k\}$ the total volume across all nodes is at most $\lambda^j V$. Since $1 < \lambda < m^{\alpha}$, summing over all levels yields
\[
  \sum_{j=0}^{k} O_{\alpha}\!\left(\lambda^j \frac{V}{Bw}\right) + O_{\alpha}(1)
  = O_{\alpha}\!\left(1 + \lambda^k \frac{V}{Bw}\right)
  = O_{\alpha}\!\left(1 + \frac{V}{Bw}\left[1 + \left(\frac{e}{t}\right)^{\alpha}\right]\right).
\]
To interchange two $u$-bit binary address fields as in Lemma~\ref{lem:memory-cutoff-interchange}, set $e = \lceil \log_{q_0} 2^u \rceil = \Theta_{\alpha}(u)$, pad both fields to $[q_0^e]$ via a sequential scan (multiplying the data volume by at most $q_0^2 = O_{\alpha}(1)$), run the $q_0$-ary interchange, and strip the padding by another scan. Since $t = \Theta_{\alpha}(a)$, the block-transfer cost of this interchange satisfies
\[
  Q_{\alpha}(u, V) = O_{\alpha}\!\left(1 + \frac{V}{Bw}\left[1 + \left(\frac{u}{a}\right)^{\alpha}\right]\right).
\]
Depth-first execution reuses buffers so internal memory stays within $M$ words, and disk space is polynomial in $N$ because each of the $O_{\alpha}(\log(2+e))$ levels has polynomial live data volume.

\paragraph{Reduction to sorting.}
Substituting $Q_{\alpha}(u, V)$ for Lemma~\ref{lem:memory-cutoff-interchange} leaves the proofs of Lemma~\ref{lem:scalar-local} and Proposition~\ref{prop:msd-sorting} unchanged: each radix level (or the single local sort when $N \le M^2$) requires $O_{\alpha}\!\left(\frac{N}{B}[1 + (h/a)^{\alpha}]\right)$ block transfers across at most $O(\max\{1, (\log N)/h\})$ levels. Using $a \le h \le \log N$ and $\alpha = 1 - \varepsilon$, multiplying these two factors gives
\[
  O_{\varepsilon}\!\left(\frac{N}{B}\left[1 + \frac{\log N}{a^{1-\varepsilon} h^{\varepsilon}}\right]\right)
  = O_{\varepsilon}\!\left(\frac{N}{B}\left[1 + \frac{\log N}{(\log(M/B))^{1-\varepsilon}(\log M)^{\varepsilon}}\right]\right).
\]
The remaining ranges of $M$ and $B$ follow from Corollary~\ref{cor:parameter-transfer} and Lemma~\ref{lem:direct-radix} exactly as in the proof of Theorem~\ref{thm:sorting}. When $M = 2B$ and $N = B^2$, setting $\varepsilon = 1/2 - \delta \in (0, 1/2)$ yields $O_{\delta}\!\left(B(\log B)^{1/2+\delta}\right)$.
\end{proof}

\section*{AI methodology}
\phantomsection
\addcontentsline{toc}{section}{AI methodology}
\label{sec:ai-methodology}

We used OpenAI's ChatGPT, including Codex, to assist with exploring network-coding constructions and sorting algorithms, developing and examining mathematical arguments, implementing and debugging Lean proofs, and preparing the manuscript and literature discussion. The tools materially affected the constructions and algorithmic analyses in Sections~\ref{sec:circuit-framework}--\ref{sec:global-sorting}, the supporting arguments in Appendices~\ref{app:io-simulation}--\ref{sec:improved-interchange}, the Lean formalization, and their exposition in Sections~\ref{sec:introduction}--\ref{sec:overview}. We also used Google's Gemini to help polish the manuscript's prose. The authors retain responsibility for the correctness, originality, and integrity of all content, including proofs, code, and references.

\FloatBarrier
\begingroup
% Keep references at the same 11 pt base size as the submission text.
\bibliographystyle{alpha}
\bibliography{references}
\endgroup
\Needspace{12\baselineskip}
\end{document}